\documentclass[11pt]{article}
\usepackage[left=1in,top=1in,right=1in,bottom=1in,head=0in,nofoot]{geometry}

\usepackage{setspace,url,bm,amsmath} 
\usepackage{scalerel}

\usepackage{xcolor}
 \usepackage{multirow}
 \usepackage{arydshln}
 \usepackage{mathtools}
\usepackage{titlesec} 
\titlelabel{\thetitle.\quad} 
\titleformat*{\section}{\bf\Large\center}

\usepackage{float}
\usepackage{graphicx} 
\usepackage{bbm}
\usepackage{latexsym}
\usepackage{caption}
\usepackage[margin=20pt]{subcaption}
\usepackage{enumerate}

\usepackage{tikz}
\usetikzlibrary{arrows.meta}

\usepackage{algorithm}
\usepackage{algorithmic}

\usepackage{amsthm}
\usepackage{amssymb}	
\usepackage{amsmath}
\usepackage{color}

\usepackage[labelformat=simple]{subcaption}

\usepackage{booktabs,caption}
\usepackage{enumitem}
\usepackage{amsfonts}
\usepackage[all,import]{xy}
\def\T{{ \mathrm{\scriptscriptstyle T} }}

\newtheorem{lemma}{Lemma}

\newtheorem{proposition}{Proposition}

\usepackage{natbib} 
\bibpunct{(}{)}{;}{a}{}{,} 

\renewcommand{\refname}{References} 
\usepackage{etoolbox} 
\apptocmd{\sloppy}{\hbadness 10000\relax}{}{} 

\usepackage{multibib}
\newcites{sec}{References}

\usepackage{color}
\usepackage{listings}
\usepackage{hyperref}
\usepackage{booktabs}

\usepackage{lscape}

\title{\bf Fairness-Aware Kidney Exchange and Kidney Paired Donation}
\author{Mingrui Zhang, Xiaowu Dai, and Lexin Li\footnote{Mingrui Zhang, Division of Biostatistics, University of California, Berkeley, CA 94720 U.S.A. (E-mail: mingrui\_zhang@berkeley.edu). Xiaowu Dai, Department of Statistics and Data Science and Department of Biostatistics, University of California, Los Angeles, CA 90095 U.S.A. (E-mail: dai@stat.ucla.edu). Lexin Li, University of California, Berkeley, CA 94720 U.S.A. (E-mail: lexinli@berkeley.edu)
}
}
\date{}
\RequirePackage[normalem]{ulem}

\allowdisplaybreaks

\begin{document}
\onehalfspacing

\newpage

 \maketitle

\begin{abstract}
The kidney paired donation (KPD) program provides an innovative solution to overcome incompatibility challenges in kidney transplants by matching incompatible donor-patient pairs and facilitating kidney exchanges. To address unequal access to transplant opportunities, there are two widely used fairness criteria: group fairness and individual fairness. However, these criteria do not consider protected patient features, which refer to characteristics legally or ethically recognized as needing protection from discrimination, such as race and gender. Motivated by the calibration principle in machine learning, we introduce a new fairness criterion: the matching outcome should be conditionally independent of the protected feature, given the sensitization level. We integrate this fairness criterion as a constraint within the KPD optimization framework and propose a computationally efficient solution using linearization
strategies and column-generation methods. Theoretically, we analyze the associated price of fairness using random graph models. Empirically, we compare our fairness criterion with group fairness and individual fairness through both simulations and a real-data example.

\end{abstract}

\section{Introduction}

\subsection{Kidney paired donation programs}

Kidney transplantation is the preferred treatment for end-stage renal disease (ESRD), offering significant improvements in both quality of life and survival compared to dialysis. However, a major obstacle is the incompatibility between donors and patients, often due to mismatches in blood type or human leukocyte antigens (HLA). According to the United Network for Organ Sharing (UNOS) and the Organ Procurement and Transplantation Network (OPTN), over 90,000 patients were on the kidney transplant waiting list at the end of 2023.

To overcome these incompatibility challenges, kidney paired donation (KPD) programs have been developed as an innovative solution. These programs match incompatible donor-patient pairs using a \textit{virtual crossmatch}, a preliminary compatibility test, and facilitate kidney exchanges, allowing patients to receive kidneys from compatible donors through mutual exchange. These exchanges are called \textit{exchange cycles} or simply \textit{cycles}, with formal definitions provided in Section \ref{sec::method}. The primary goal of KPD programs is to maximize the number of successful transplants or optimize generalized utilities based on predicted transplantation or survival outcomes. 

Despite their promise, planned cycles may fail for various reasons, such as illness, pregnancy, or the death of a patient or donor, scheduling conflicts, or discrepancies between \textit{virtual} and \textit{laboratory crossmatch} results. These uncertainties necessitate the consideration of \textit{recourse} strategies, which identify alternative transplant opportunities within the original cycle. Even when an entire cycle cannot proceed, smaller unaffected sub-cycles may still be viable. By accounting for these uncertainties, KPD programs can further maximize the expected number of successful transplants or optimize the expectation of some general utilities. 

Following the framework of \cite{klimentova2016maximising}, KPD programs adopt three main recourse strategies. The first is \textit{no recourse}, where a cycle either proceeds fully or fails entirely \citep[e.g.][]{li2014optimal}. The second, \textit{internal recourse}, identifies the sub-cycle with the highest utility among those unaffected by the failure \citep[e.g.][]{pedroso2014maximizing}. The third, \textit{subset recourse}, considers broader subsets that may include multiple cycles, enabling alternative arrangements when uncertainties arise \citep[e.g.][]{bray2018valuing}. Each strategy computes expected utilities accordingly, with KPD programs aiming to maximize these expected utilities. 

In addition to recourse strategies, we can integrate uncertainties into the optimization framework through other approaches, such as look-ahead strategy \citep{wang2017look} and robust optimization \citep{mcelfresh2019scalable}. 

\subsection{Fairness concerns in KPD programs}

Despite the success of KPD programs, important fairness concerns arise, particularly regarding unequal access to transplant opportunities. These disparities stem from two main factors. The first factor is differences in patients' HLA sensitization levels. A patient's HLA sensitization level is measured by their panel-reactive antibody (PRA) score, which reflects the likelihood of HLA incompatibility with a random donor. Patients with high PRA scores, referred to as highly sensitized patients, are more difficult to match and, therefore, have fewer transplant opportunities. In contrast, patients with low PRA scores, known as lowly-sensitized patients, are easier to match and benefit from more transplant opportunities. The second factor is asymmetric blood type compatibility. According to standard ABO blood type compatibility rules, patients with blood type O are harder to match because they can only receive kidneys from O-type donors. Conversely, patients with blood type AB are the easiest to match, as they can receive kidneys from donors of any blood type. 

To address these disparities, KPD programs can incorporate fairness constraints to reduce unfairness caused by differences in HLA sensitization and blood type. Two widely used fairness criteria in KPD programs are group fairness and individual fairness. Group fairness focuses on ensuring that highly-sensitized patients receive equitable consideration relative to lowly-sensitized patients \citep{dickerson2014price, mcelfresh2019scalable, freedman2020adapting}. In contrast, individual fairness aims to provide balanced selection chances for each patient, ensuring that no one is unfairly disadvantaged \citep{farnadi2021individual}. While these two fairness criteria are central to addressing disparities in kidney exchange, other approaches have also been explored. For example, \cite{st2022adaptation} incorporate the Nash standard of comparison (or proportional fairness) and Rawlsian justice principles. \cite{ashlagi2014free}, \cite{klimentova2021fairness}, and \cite{carvalho2023theoretical} draw on game theory to address fairness within utility-maximization frameworks, considering the interests of stakeholders such as hospitals and regions. 

In the context of fair machine learning, a \textit{protected feature} (or sensitive attribute) refers to a characteristic legally or ethically recognized as needing protection from discrimination. Our key question is how to establish a fairness criterion that ensures equal access to transplant opportunities across patient groups defined by protected characteristics---and how to achieve this in practice. This specific focus has not yet been explored in the KPD literature, as existing fairness criteria do not account for protected patient features.

Some protected features, such as race and gender, are associated with differences in sensitization levels and blood types. For instance, studies have shown that parous women are more likely to develop high sensitization to HLA antigens \citep{bromberger2017pregnancy}, making them less compatible with most donors in a KPD program and harder to match. This leads to unequal access to transplant opportunities between females and males. A simplistic approach might aim to balance overall selection rates between genders, but this could inadvertently disadvantage highly-sensitized male patients, as a higher number of highly-sensitized female patients would need to be matched to achieve gender balance. A more equitable approach would balance selection rates within subgroups, such as highly-sensitized females versus males and lowly-sensitized females versus males.

Motivated by this example, we propose a new fairness criterion: the matching outcome should be conditionally independent of the protected feature, given the sensitization level. The randomness associated with this fairness criterion is determined by a randomization policy, as proposed for individual fairness in kidney exchange \citep{farnadi2021individual, st2022adaptation} and general matching problems \citep{garcia2020fair, karni2021fairness}. This approach provides guarantees for average selection rates within protected groups across each sensitization level.

\subsection{Fairness in general decision-making problems}
Our fairness criterion in KPD programs is defined based on the conditional outcome given protected features, drawing on similar concepts from the literature on general decision-making problems. 

\textit{Demographic parity} ensures fairness by requiring that the rate of positive decisions is consistent across groups defined by protected features, promoting equality in outcomes regardless of group membership. \textit{Equalized odds}, introduced by \cite{hardt2016equality}, aligns predictive performance such that the false positive and false negative rates are similar across groups, leading to a fair distribution of errors. \textit{Predictive parity}, discussed by \cite{chouldechova2017fair}, ensures fairness by equalizing the positive predictive value (PPV) across groups, thereby making positive predictions equally reliable and trustworthy for all groups. \textit{Calibration within groups}, explored by \cite{kleinberg_et_al:LIPIcs.ITCS.2017.43}, requires that individuals with the same predicted probability have consistent actual outcome rates across groups, ensuring well-calibrated predictions.

These fairness concepts emphasize different priorities: overall outcome equality, error distribution, prediction reliability, or probability calibration. They are frequently incorporated as fairness constraints in statistical optimization problems \citep[e.g.,][]{liebl2023fast}. 

Our fairness criterion in KPD programs is closely aligned with calibration within groups; see a more detailed discussion in Section \ref{subsec::fairness_definition}.

\subsection{Our contributions}\label{subsec::contributions}

This paper makes several contributions to the field of kidney exchange and fairness in allocation. First, we propose a new fairness criterion based on a protected feature, which has not been explored in the kidney exchange literature. We integrate this fairness criterion as a constraint within the optimization framework commonly used in kidney paired donation (KPD) programs. This flexible structure can accommodate other fairness criteria and potential recourse strategies. Furthermore, we propose a computationally efficient solution to the resulting optimization problem, achieving scalability by avoiding explicit enumeration of all exchange plans and instead relying on an iterative cutting-plane procedure.

Second, we investigate the \textit{price of fairness} associated with our proposed criterion, defined as the relative loss in system efficiency when a fair allocation is prioritized over an optimal (unconstrained) allocation \citep{bertsimas2011price}. Theoretically, we derive an upper bound on the asymptotic price of fairness using random graph models that incorporate ABO blood type distributions. Empirically, through simulation studies, we show that the efficiency loss from implementing our fairness criterion is relatively low. 

Our findings align with prior studies on the tradeoff between efficiency and fairness in resource allocation. For example, \cite{dickerson2014price} examine the efficiency loss associated with group fairness in kidney exchange, while \cite{ashlagi2014free} analyze the price of ensuring individual rationality in multi-hospital kidney exchanges, both employing random graph models with ABO blood types. \cite{st2022adaptation} utilize the Nash Social Welfare Program to address the tradeoff between fairness and efficiency. Similarly, \cite{viviano2024fair} propose a framework for fair policy targeting that balances fairness and efficiency using Pareto optimal treatment allocation rules, offering theoretical guarantees and practical solutions applicable to social welfare contexts.

\section{Review of KPD program in an optimization framework}\label{sec::method}

In this section, we review the notations, terminologies, and optimization framework used in KPD programs. Section \ref{subsec::classical} focuses on the classical optimization problem without incorporating fairness constraints, while Section \ref{subsec::fairness-review} reviews group fairness and individual fairness, introducing an additional constraint to integrate fairness into the framework. 

\subsection{KPD program without fairness}\label{subsec::classical}

\subsubsection{Classical formulation}

We represent a KPD program as a directed graph $G=(V,E)$, where $V=\{v_1,\cdots ,v_N\}$ denotes the vertex set and $E$ denotes the edge set. The vertex set $V$ is the set of $N$ incompatible donor-patient pairs. The edge set $E$ consists of all ordered pairs $(v_i,v_j)$ that the donor in vertex $v_i\in V$ is compatible with the patient in vertex $v_j\in V$. An \textit{exchange cycle}, or simply \textit{cycle}, is defined as a sequence of distinct incompatible donor-patient pairs. We denote a cycle $c$ as an ordered sequence of vertices $[c_1,\cdots ,c_{|c|}]$ in $V$, where $|c|$ is the cycle length of $c$, satisfying that $(c_i,c_{i+1})\in E$ for $1\leq i\leq |c|-1$, and $(c_{|c|},c_1)\in E$. To execute the cycle, the patient in $c_{i+1}$ will receive the donor kidney of $c_i$ for $1\leq i\leq |c|-1$, and the patient $c_1$ will receive the donor kidney of $c_{|c|}$. Figure \ref{fig::cycles} illustrates how a two-way cycle and a three-way cycle work in KPD programs. An \textit{exchange plan} is a collection of vertex-disjoint cycles in the graph. 

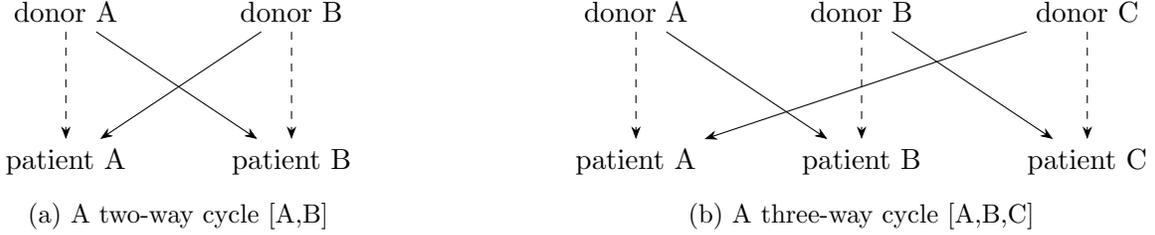
\begin{figure}
    \centering
    \begin{subfigure}[b]{0.35\textwidth}
        \centering
\begin{tikzpicture}

    \node (donorA) at (0, 2) {donor A};
    \node (donorB) at (3, 2) {donor B};
    \node (patientA) at (0, 0) {patient A};
    \node (patientB) at (3, 0) {patient B};

    \draw[dashed, -{Stealth}] (donorA) -- (patientA);
    \draw[dashed, -{Stealth}] (donorB) -- (patientB);

    \draw[-{Stealth}] (donorA) -- (patientB);
    \draw[-{Stealth}] (donorB) -- (patientA);

\end{tikzpicture}
        \caption{A two-way cycle [A,B]}
    \end{subfigure}
    \hfill
    \begin{subfigure}[b]{0.55\textwidth}
        \centering
\begin{tikzpicture}

    \node (donorA) at (0, 2) {donor A};
    \node (donorB) at (3, 2) {donor B};
    \node (donorC) at (6, 2) {donor C};
    
    \node (patientA) at (0, 0) {patient A};
    \node (patientB) at (3, 0) {patient B};
    \node (patientC) at (6, 0) {patient C};

    \draw[dashed, -{Stealth}] (donorA) -- (patientA);
    \draw[dashed, -{Stealth}] (donorB) -- (patientB);
    \draw[dashed, -{Stealth}] (donorC) -- (patientC);

    \draw[-{Stealth}] (donorA) -- (patientB);
    \draw[-{Stealth}] (donorB) -- (patientC);
    \draw[-{Stealth}] (donorC) -- (patientA);

\end{tikzpicture}
        \caption{A three-way cycle [A,B,C]}
    \end{subfigure}
    
    \caption{Illustration of exchange cycles shown with solid arrows. Transplantations along the dashed arrows cannot proceed due to incompatibility. }\label{fig::cycles}
\end{figure}

In fielded KPD programs, we usually restrict cycles to be no longer than 3. We use $\mathcal{C}$ to denote the set of all cycles under the restrictions on the length of the cycles. We assign a utility $u_{v_i,v_j}$ to each edge $(v_i,v_j)\in E$. Based on the edge utility, we define the cycle utility $u(c)=\sum_{i=1}^{|c|-1}u_{c_i,c_{i+1}}+u_{c_{|c|},c_1}$. In a classical KPD program, we aim to select a disjoint packing of cycles in $\mathcal{C}$ with the maximum sum of utilities, which can be formulated as the following integer program: 
\begin{equation}
\label{eq::cycle}
\max_{x_c\in\{0,1\}}\quad \sum_{c\in \mathcal{C}}x_cu(c) \quad \text{subject to}\quad \sum_{c\in\mathcal{C}}x_c1_{v\in c}\leq 1,\ \forall\ v\in V.
\end{equation}
The binary variable $x_c$ indicates whether $c$ is selected or not. The objective function in \eqref{eq::cycle} is the sum of utilities of the selected cycles. The constraint in \eqref{eq::cycle} requires that the selected cycles are vertex-disjoint. 

\subsubsection{Extension: incorporating recourse strategies}\label{subsec::extension}

To account for uncertainty, we consider the expected utilities in \eqref{eq::cycle}, depending on the chosen recourse strategy. Different strategies lead to different calculations of expected utilities. 

\paragraph{No-recourse strategy.} 
Under the no-recourse strategy, the expected utility is determined without adapting to failures. Given the failure probabilities of each vertex in \( V \) and each edge in \( E \), we can compute the probability of each cycle in \( \mathcal{C} \) being executable, where all vertices and all edges in the cycle do not fail. An explicit formula is available to compute the expected utility for this strategy \citep{li2014optimal, klimentova2016maximising}.

\paragraph{Internal-recourse strategy.}
The internal-recourse strategy considers adapting to failures within a given cycle. For a cycle \( c \), let \( \mathcal{M}(c) \) represent the set of all exchange plans involving vertices in \( c \). Without loss of generality, assume the elements of \( \mathcal{M}(c) \) are ordered by descending utility. Specifically, the utility of \( M_i(c) \), the \( i \)-th element in \( \mathcal{M}(c) \), is no less than that of \( M_j(c) \), the \( j \)-th element, for \( i < j \). If all vertices and edges in \( M_1(c) \) do not fail, \( M_1(c) \) will be executed. If any vertices or edges in \( M_k(c) \) fail but those in \( M_{k+1}(c) \) do not fail, \( M_{k+1}(c) \) will be executed (\( k \geq 1 \)). While there is no explicit formula for this strategy, efficient algorithms are available for computing the expected utility \citep{li2014optimal, pedroso2014maximizing, klimentova2016maximising}.

\paragraph{Subset-recourse strategy.}
The subset-recourse strategy expands the search space to disjoint \textit{relevant subsets} rather than disjoint cycles, providing more flexibility in adapting to uncertainty. Formally, a relevant subset of size \((k, q)\) is the set of at most \((k+q)\) vertices in graph $G$ inducing a strongly connected subgraph, such that any edge of the paths that provide the strong connectivity belongs to some cycle of size at most $k$, whose vertices are in the subset. We follow the same approach as the internal-recourse strategy to compute the expected utility of each relevant subset. To solve the optimization problem in \eqref{eq::cycle}, we must enumerate the new set $\mathcal{C}$, with algorithms given in \cite{klimentova2016maximising} and \cite{wang2019efficient}.

\subsection{KPD program with fairness constraint}\label{subsec::fairness-review}

\subsubsection{Group fairness}\label{subsec::group}

The group fairness aims to ensure fairness towards highly-sensitized patients. We partition the vertex set $V$ into $V_h\cup V_e$, where $V_h$ denotes the set of incompatible donor-patient pairs with highly-sensitized patients, and $V_e$ denotes the set of incompatible donor-patient pairs with lowly-sensitized patients. Following \cite{dickerson2014price}, we formulate the group fairness as a constraint in the optimization problem below
\begin{equation}
\label{eq::group}
\max_{x_c\in\{0,1\}}\  \sum_{c\in \mathcal{C}}x_cu(c)\quad \text{subject to}\quad \sum_{v\in V_h}\sum_{c\in \mathcal{C}}x_c1_{v\in c}\geq \alpha,\  \sum_{c\in \mathcal{C}}x_c1_{v\in c}\leq 1,\ \forall\ v\in V.
\end{equation}
Specifically, we consider a fairness constraint that the number of highly-sensitized patients involved in the matching is no less than some threshold $\alpha$.

\subsubsection{Individual fairness}\label{subsec::individual}
The individual fairness utilizes a randomization policy to ensure that every patient has a similar chance of being matched. Let $\mathcal{F}\subseteq 2^{\mathcal{C}}$ be the set of possible exchange plans that consist of disjoint cycles in $\mathcal{C}$. We assign a probability distribution $\delta$ on $\mathcal{F}$, which indicates the probability of selecting each exchange plan in $\mathcal{F}$. Based on $\delta$, let $x_c$ be the overall probability that the cycle $c\in \mathcal{C}$ is selected as one of the disjoint subsets. Further, we can compute $\pi_v=\sum_{c\in\mathcal{C}}x_c1_{v\in c}$, which is the probability of each vertex $v$ to be selected. Let $\pi=(\pi_v)_{v\in V}\in\mathbb{R}^{N}$ be the vector of selection probability of all patients and let $\overline{\pi}=1^{\T}\pi/N$ be the average selection probability of all patients. The quantity $\pi_v$ is central in defining individual fairness and our new fairness criterion in Section \ref{sec::new}. 

Recall that $x_c$ in \eqref{eq::cycle} and \eqref{eq::group} is a binary variable indicating whether $c$ is selected or not, but here $x_c$ is a continuous variable bounded between 0 and 1. Following \cite{farnadi2021individual}, we formulate the individual fairness as a constraint in the optimization problem below
\begin{equation}
\label{eq::individual}
\max_{\delta}\quad \sum_{c\in \mathcal{C}}x_cu(c)\quad \text{subject to}\quad \sum_{v\in V}|\pi_v-\overline{\pi}|^p\geq \beta^p.
\end{equation}
Specifically, the individual fairness promotes a similar chance of being selected for each patient. We choose the $L_p$ norm of the vector $(\pi-\overline{\pi})$ to measure the variation of the selection probability among all the patients. We consider a fairness constraint that the $L_p$ norm of the vector $(\pi_v-\overline{\pi})$ is no less than some threshold $\beta$.

\section{A new fairness criterion based on a protected feature}\label{sec::new}

\subsection{Fairness formulation}\label{subsec::fairness_definition}

In this section, we introduce a new fairness criterion based on a protected feature $A$ of patients in $V$, motivated by the concept of \textit{calibration within group} in the context of machine learning. In essence, our new fairness criterion requires that the probability of a patient being selected for a transplant (the matching outcome) should not depend on the protected feature, such as race or gender, once we account for the patient’s sensitization level. This means that, within each sensitization category, patients with different protected attributes should have similar chances of being matched.

In binary classification, a score function $R(X)$ satisfies calibration within groups if $\mathrm{Pr}(Y=1\mid R(X)=r, A=a)=r$ for all score values $r$ and group level $a$, where $Y$ is the outcome variable and $X$ is the feature variables. A slightly weaker definition \citep{corbett2023measure} only requires that $Y$ is conditionally independent of $A$ given the score $R(X)$. Our fairness definition is based on this weaker definition. In KPD programs, we view $Y$ as the selection indicator and $R$ as the sensitization level of a patient. Therefore, we define our new fairness criterion as satisfying that the selection indicator is conditionally independent of the protected feature given the sensitization level. In other words, at each sensitization level, the matching outcome is independent of the protected feature. We can view this fairness criterion in KPD programs as a reverse problem of that in machine learning. In machine learning, the randomness is due to the underlying population model, and our goal is to construct a score $R(X)$ satisfying the fairness condition. In KPD programs, we need to determine the randomness such that the observed sensitization level satisfies the fairness condition. 

For simplicity of presentation, we assume the protected feature $A$ is binary with two levels $\{0,1\}$ and the sensitization $R$ has $M$ levels $\{r_1,r_2,\cdots ,r_M\}$. In practice, it is common that $M=2$ where we partition all patients into highly-sensitized and lowly-sensitized patients, or $M=3$ where we partition all patients into highly-sensitized, moderately-sensitized and lowly-sensitized patients. The variables $A$ and $R$ partition all patients into $2M$ subgroups, denoted as $\{V_{ij}\}_{0\leq i\leq 1,1\leq j\leq M}$, where $V_{ij}=\{v\in V: A(v)=i, R(v)=r_j\}$. As in Section \ref{subsec::individual}, we assign a probability distribution $\delta$ on the set of exchange plans $\mathcal{F}$, which indicates the probability of selecting each exchange plan in $\mathcal{F}$. Moreover, $\pi_v$ is the probability of each vertex $v$ to be selected in the end, and thus $|V_{ij}|^{-1}\sum_{v\in V_{ij}}\pi_v$ is average selection rate in $V_{ij}$. Given the sensitization level $R=r_j$, our fairness criterion restricts the average selection rates in $V_{0j}$ and $V_{1j}$ to be close. It is natural to impose constraints that the absolute difference in these average selection rates is bound by a constant $l_j$. Therefore, we formulate the problem under our fairness criterion as
\begin{equation}
\label{eq::fair-0}
\max_{\delta}\ \sum_{c\in \mathcal{C}}x_cu(c) \quad \text{subject to}\ \left| |V_{0j}|^{-1}\sum_{v\in V_{0j}}\pi_v-|V_{1j}|^{-1}\sum_{v\in V_{1j}}\pi_v\right| \leq l_j,\ \forall 1\leq j\leq M.
\end{equation}

For whichever choice of parameters $l_1,\cdots ,l_M\geq 0$, there always exists a solution to \eqref{eq::fair-0}. The strength of the fairness constraint depends on the parameters $l_1,\cdots l_M$. In practice, we suggest two candidate values for $l$'s: $l_j=0.5/\min\{|V_{1j}|, |V_{2j}|\}$ and $l_j=0.5/\max\{|V_{1j}|, |V_{2j}|\}$. We can interpret these two candidate values as the desired precision based on the larger or smaller subgroup. The former represents a weaker fairness constraint, and the latter represents a stronger one.

\subsection{Algorithm for solving \eqref{eq::fair-0}}
\label{subsec:algo_fair}

This subsection presents a linear optimization procedure for solving \eqref{eq::fair-0} without enumerating all exchange plans. The approach follows a master-subproblem structure similar to column generation. The master problem optimizes over the marginal probabilities of selecting each cycle and each vertex, while an auxiliary subproblem verifies whether the proposed solution is implementable and, if necessary, adds new linear constraints that refine the master problem. 

Specifically, we express \eqref{eq::fair-0} as an optimization problem in terms of the cycle-selection variables $x_c$, rather than the probability distribution $\delta$ over exchange plans. Recall that the selection probability of each vertex $v$ is 
\[
\pi_v = \sum_{c\in\mathcal{C}} 1_{v\in c}\,x_c.
\]
Both the objective and the fairness constraints are linear in $x_c$. A direct constraint of variables $x_c$ is
$0\leq x_c\leq 1$ for all $c\in\mathcal{C}$ and $0\leq \pi_v\leq 1$ for all $v\in V$. However, not every such vector $x=(x_c)_{c\in\mathcal{C}}$ satisfying the above constraint can be realized by a valid probability distribution over disjoint exchange plans.

To enforce implementability, we introduce a scalar variable $\eta\in\mathbb{R}$ that represents the best achievable expected utility at the current selection profile and iteratively add linear inequalities, called \emph{cuts}, to refine the feasible region. Let $\mathcal{B}^{\mathrm{opt}}$ and $\mathcal{B}^{\mathrm{feas}}$ denote the sets of indices for optimality and feasibility cuts generated so far; both sets start empty and expand monotonically. The restricted master problem is formulated as
\begin{equation}
\begin{aligned}
\max_{x,\pi,\eta}\quad & \eta\\
\text{subject to}\quad 
& \pi_v=\sum_{c\in\mathcal{C}}1_{v\in c}x_c\leq 1, \qquad \forall v\in V,\\
& 0\leq x_c\leq 1,\qquad \forall c\in\mathcal{C},\\
& \left|\,|V_{0j}|^{-1}\sum_{v\in V_{0j}}\pi_v - |V_{1j}|^{-1}\sum_{v\in V_{1j}}\pi_v\,\right| \le l_j,\qquad j=1,\ldots,M,\\
& \pi^{\T}\theta^{(\ell)} + \beta^{(\ell)}\geq \eta, \qquad \forall\,\ell\in\mathcal{B}^{\mathrm{opt}},\\
& \pi^{\T}\hat{\theta}^{(k)}+\hat{\beta}^{(k)}\leq 0, \qquad \forall\,k\in\mathcal{B}^{\mathrm{feas}},
\end{aligned}
\tag{MP}
\end{equation}
where $(\theta^{(\ell)},\beta^{(\ell)})\in\mathbb{R}^{N}\times\mathbb{R}$ and $(\hat{\theta}^{(k)},\hat{\beta}^{(k)})\in\mathbb{R}^{N}\times\mathbb{R}$ are coefficients provided by the auxiliary subproblem. 

Given a candidate $\pi$ from \textnormal{(MP)}, the auxiliary subproblem checks whether $\pi$ can be implemented by a valid distribution over exchange plans and determines the best achievable expected utility under that profile. For each plan $F_t\in\mathcal{F}$, let $u(F_t)$ denote its total utility and let $\delta(F_t)$ denote the probability that plan $F_t$ is selected. The auxiliary configuration problem is
\begin{equation}
\begin{aligned}
\max_{\delta}\quad & \sum_{F_t\in\mathcal{F}} u(F_t)\, \delta(F_t)\\
\text{subject to}\quad & \sum_{F_t\in\mathcal{F}} 1_{v\in F_t}\,\delta(F_t) = \pi_v, \qquad \forall v\in V,\\
& \sum_{F_t\in\mathcal{F}} \delta(F_t) = 1,\qquad \delta(F_t) \ge 0, \quad \forall F_t\in\mathcal{F},
\end{aligned}
\tag{SP}
\end{equation}
with the corresponding dual formulation
\begin{equation}
\begin{aligned}
\min_{\theta\in\mathbb{R}^N,\beta\in\mathbb{R}}\quad & \pi^{\T}\theta + \beta\\
\text{subject to}\quad & \beta + \sum_{v\in V}\theta_v\, 1_{v\in F_t} \ge u(F_t), \qquad \forall F_t\in\mathcal{F}.
\end{aligned}
\tag{SP--D}
\end{equation}

In practice, we solve the dual problem \textnormal{(SP--D)} rather than the primal \textnormal{(SP)}, since the primal has one variable for each exchange plan and cannot be handled directly. The dual problem is of manageable size and can be solved efficiently using column generation. The feasibility of \textnormal{(SP)} is inferred from the behavior of \textnormal{(SP--D)}: if the dual is bounded, then \textnormal{(SP)} is feasible, and if the dual is unbounded, then \textnormal{(SP)} is infeasible.

If \textnormal{(SP)} is feasible, strong duality ensures that the optimal values of \textnormal{(SP)} and \textnormal{(SP--D)} coincide. The resulting optimal dual pair $(\theta,\beta)$ defines a valid optimality cut $\pi^{\T}\theta + \beta \ge \eta$, which is added to \textnormal{(MP)}. If \textnormal{(SP)} is infeasible, then, by Farkas’ lemma, there exists a certificate $(\hat{\theta},\hat{\beta})$ such that $\hat{\beta}+\sum_{v\in V}\hat{\theta}_v 1_{v\in F_t}\le 0$ for all $F_t\in\mathcal{F}$ and $\pi^{\T}\hat{\theta}+\hat{\beta}>0$. The separating inequality $\pi^{\T}\hat{\theta}+\hat{\beta}\le 0$ is then added as a feasibility cut. 

The dual problem \textnormal{(SP--D)} contains one constraint for each feasible plan $F_t$, which makes it impossible to construct explicitly. Column generation avoids this by starting with a small subset of constraints and iteratively adding only those that are necessary for optimality. At each iteration, the algorithm identifies the most violated dual constraint under the current dual solution $(\theta,\beta)$.

Formally, given $(\theta,\beta)$, we search for a plan $F_t$ whose dual constraint
\[
\beta + \sum_{v\in V}\theta_v\, 1_{v\in F_t} \ge u(F_t)
\]
is most violated. This search is equivalent to solving a maximum-weight packing problem over vertex-disjoint cycles:
\[
\max_{y_c\in\{0,1\}}\ 
\sum_{c\in\mathcal{C}}\Bigl(u(c)-\sum_{v\in c}\theta_v\Bigr)y_c - \beta
\quad
\text{subject to}\quad 
\sum_{c\in\mathcal{C}} y_c\,1_{v\in c} \le 1,\quad \forall v\in V.
\]
Each binary variable $y_c$ indicates whether cycle $c$ is selected, and the constraint ensures that selected cycles are vertex-disjoint. If the optimal value of this problem is positive, the corresponding plan violates a dual constraint and is added to a restricted version of \textnormal{(SP--D)}. If the optimal value is nonpositive, all constraints in the full dual problem are satisfied. Because $(\theta,\beta)$ is already optimal for the restricted dual problem, and adding satisfied constraints cannot further decrease its objective value, $(\theta,\beta)$ must also be optimal for the full dual \textnormal{(SP--D)}.

The algorithm alternates between solving the master problem \textnormal{(MP)} and invoking the auxiliary subproblem. Each iteration produces a candidate $(x,\pi,\eta)$ that satisfies the fairness constraints in \eqref{eq::fair-0}. The auxiliary subproblem then determines whether this profile $\pi$ is implementable. If it is, the resulting dual solution provides an optimality cut that tightens the objective bound $\eta$ in \textnormal{(MP)}; otherwise, a feasibility cut removes the unattainable profile $\pi$ from the feasible region.

The process continues until no violated constraint can be found by the auxiliary subproblem. In practice, this occurs when the optimal value of the separation problem in the column-generation step is nonpositive, indicating that all dual constraints are satisfied. At this point, the restricted master problem is equivalent to the full master problem, and no further cuts or columns can improve the solution. 

From a theoretical standpoint, each iteration introduces a valid linear inequality that progressively refines the feasible region, and only finitely many distinct cuts can be generated since the set of feasible plans $\mathcal{F}$ is finite. Hence, the procedure terminates after a finite number of iterations. Upon termination, the current $\pi$ is guaranteed to be implementable by a valid distribution over exchange plans, the fairness constraints are satisfied, and the corresponding objective value $\eta$ equals the optimum of \eqref{eq::fair-0}. The final primal solution of \textnormal{(SP)} explicitly provides this distribution, yielding an optimal randomized policy that achieves both fairness and efficiency.

\subsection{Prediction of individual selection probability}\label{subsec::estimation}

The selection probability $\pi_v$ is a central quantity in defining both individual fairness and our new fairness criterion. Since a KPD program is not static but dynamic, a natural statistical question to ask is how to predict individual selection probabilities before more incompatible donor-patient pairs enter the pool for the next round of exchange allocation. We provide a solution based on sample splitting. The following discussion can accommodate any fairness criteria within our general framework.

Our method utilizes historical data of incompatible donor-patient pairs independent of the current and future pairs. Suppose the historical pool consists of incompatible donor-patient pairs denoted as $\{\tilde{v}_1,\cdots ,\tilde{v}_{N_0}\}$, and the current pool consists of pairs denoted as $\{v_1,\cdots ,v_{N_1}\}$. Assume that the exchange allocation occurs when the size of the incompatible pairs pool reaches $N$, where $N_1 < N < N_0 + N_1$. The historical pool $\{\tilde{v}_1,\cdots, \tilde{v}_{N_0}\}$ aggregates incompatible donor-patient pairs collected across potentially multiple past KPD rounds rather than a single pool. In each iteration $b$, we sample a subset of these historical vertices to simulate a feasible past pool of size $N-N_1$. The edge set $E^{(b)}$ is then determined over the combined vertex set $V^{(b)}=\{v_1,\cdots,v_{N_1}\}\cup \tilde{V}^{(b)}$, treating unobserved edges between vertices that never coexisted in historical rounds as absent. To construct $E^{(b)}$, we recompute donor-patient compatibilities among all vertices in $V^{(b)}$ using the same ABO and HLA matching criteria as in the real KPD graph. Each edge $(v_i,v_j)$ is included in $E^{(b)}$ if donor $v_i$ is compatible with patient $v_j$. The prediction procedure can be described in Algorithm \ref{alg1} below. 

\begin{algorithm}[H]
\caption{Prediction procedure with selection probabilities}\label{alg1}
\begin{algorithmic}[1]
\STATE \textbf{Input:} Historical pool of vertex set $\{\tilde{v}_1, \dots, \tilde{v}_{N_0}\}$, current pool of vertex set $\{v_1, \dots, v_{N_1}\}$, and number of repetitions $B$.
\STATE \textbf{Output:} Prediction of $\pi_v$ for each $v\in \{v_1, \dots, v_{N_1}\}$.

\FOR{$b = 1$ to $B$}
    \STATE Sample $\{\tilde{v}_1^{(b)}, \dots, \tilde{v}_{N-N_1}^{(b)}\}$ from $\{\tilde{v}_1, \dots, \tilde{v}_{N_0}\}$ without replacement.
    \STATE Determine the edge set $E^{(b)}$ in $V^{(b)}=\{v_1, \dots, v_{N_1}, \tilde{v}_1^{(b)}, \dots, \tilde{v}_{N-N_1}^{(b)}\}$.
    \STATE Based on the graph $(V^{(b)},E^{(b)})$, solve \eqref{eq::fair-0} to obtain the selection probability of $\pi_v^{(b)}$ for each $v\in \{v_1, \dots, v_{N_1}\}$.
\ENDFOR

\STATE \textbf{Return}: mean and quantiles of $\{\pi_v^{(1)},\cdots ,\pi_v^{(B)}\}$ as the mean prediction and interval prediction of $\pi_v$, for each $v\in \{v_1, \dots, v_{N_1}\}$. 

\end{algorithmic}
\end{algorithm}

In practice, the computational cost can be high if we directly enumerate all cycles or relevant subsets within the simulated pool $\{v_1,\cdots,v_{N_1},\tilde{v}_1^{(b)},\cdots,\tilde{v}_{N-N_1}^{(b)}\}$ for each of the $B$ replications (e.g., $B=1000$). When the number of protected-feature groups $M$ is small, it is more efficient to first enumerate the feasible cycles or subsets once using the entire combined vertex set $\{v_1,\cdots,v_{N_1},\tilde{v}_1,\cdots,\tilde{v}_{N_0}\}$, and then, for each replication $b$, retain only those cycles whose vertices appear in the sampled subset $\{v_1,\cdots,v_{N_1},\tilde{v}_1^{(b)},\cdots,\tilde{v}_{N-N_1}^{(b)}\}$. This strategy avoids repeated enumeration while preserving correctness. When $N_0$ is extremely large, however, full enumeration over all historical vertices may still be infeasible. In such cases, the historical pool $\{\tilde{v}_1,\cdots,\tilde{v}_{N_0}\}$ can be partitioned into several disjoint subsets of manageable size, and the enumeration of feasible cycles or relevant subsets can be performed separately within each subset. The results are then combined across subsets as needed.

\section{Price of our new fairness criterion under random graph models}\label{sec::theory}

The \textit{price of our new fairness criterion} quantifies the efficiency loss incurred when enforcing fairness constraints compared to the unconstrained optimal allocation. Formally, it is defined as the relative reduction in the total utility (e.g., expected number of transplants or any general utility) achieved under our fairness-constrained optimization versus the utility achieved by the unconstrained solution. This concept follows the standard ``price of fairness'' framework in the literature \citep[e.g.,][]{bertsimas2011price}, here adapted specifically to our new fairness criterion. In this section, we establish theoretical guarantees for the price of our fairness criterion using a random graph model that incorporates ABO blood types. We describe the model assumptions below.

Recall that a donor and a patient are compatible if they match in both blood type and HLA. We assume blood type compatibility follows standard medical guidelines: AB patients can receive kidneys from donors of any blood type, A and B patients can receive from donors of their own type or type O, while O patients can only receive from type O donors. We assume HLA compatibility follows binomial distributions. Specifically, we randomly assign each patient a PRA score, representing the probability of being HLA incompatible with any donor. We assume the PRA scores can take discrete values $\{r_1, \dots, r_M\}$, which also determine sensitization levels $\{r_1, \dots, r_M\}$; and we assume HLA compatibility between different donor-patient pairs is independent. The vertex set $V$ is formed by independently drawing donor-patient pairs from an underlying population, keeping only incompatible pairs until a total of $N$ incompatible pairs is reached. Let $\mu_{b_1,b_2,r,a}$ denote the probability of sampling an incompatible donor-patient pair with donor blood type $b_1$, patient blood type $b_2$, patient sensitization level $r$, and patient protected group level $a$.  Here, the parameters $\{\mu_{b_1,b_2,r,a}\}$ are fixed but unknown. The edge set $E$ is determined by the compatibility between donor-patient pairs in $V$. Regarding the utility assignment, we assume $u(c)$ only depends on the induced subgraph of vertices in $c$, independent of the protected feature of vertices in $c$. We also assume $u(c)/|c|$ is bounded by a constant. 

In our random graph model, the PRA levels $\{r_1, \ldots, r_M\}$ are treated as fixed parameters representing different sensitization categories rather than random draws from a continuous distribution. Specifically, each vertex $v_i$ (representing a donor-patient pair) is assigned a PRA level $r_i$ from this finite set according to empirical proportions observed in real KPD data \citep[e.g.,][]{saidman2006increasing}. Hence, the PRA levels are not sampled uniformly over $[0,1]$ but follow a categorical distribution reflecting realistic sensitization patterns.

In our theoretical analysis, we adopt a simplified random graph model where each edge is independently present with a constant probability $p \in (0,1)$.
This assumption is primarily technical and facilitates the use of classical random graph results \citep[e.g.,][]{erdHos1968random} to establish asymptotic matching properties.
In real kidney exchange programs, however, the probability that an edge exists between two vertices depends on medical compatibility, particularly blood type and PRA sensitization level, so that the effective edge probabilities are heterogeneous across vertex pairs.
Nevertheless, the assumption of a constant $p$ captures the key structural property required in our proofs: the existence of a sufficiently dense graph where each vertex has an expected number of compatible partners growing linearly in $N$.
Under any model where the minimal edge probability remains bounded away from zero across vertex types (as is empirically the case in our simulations), the asymptotic results and efficiency bounds derived here continue to hold qualitatively.
Therefore, while the assumption of a constant $p$ simplifies the theoretical exposition, it does not materially affect the validity of the conclusions regarding the asymptotic price of fairness.

Similar to group fairness that prioritizes highly-sensitized patients, our new fairness criterion needs to prioritize some subgroups of patients based on the protected feature, and further balance their average selection probabilities. In Section \ref{subsec::general}, we derive a general result quantifying the efficiency loss due to subgroup prioritization in KPD programs. Specifically, we focus on prioritizing patients with either $A=1$ or $A=0$, given their blood type and sensitization level. In Section \ref{subsec::application}, we apply the result to establish theoretical guarantees for the price of fairness. 

\subsection{Efficiency loss of subgroup prioritization}\label{subsec::general}

Recall the definition of $V_{ij}$ in Section \ref{subsec::fairness_definition}. We further write $V_{ij}$ as the union 
$$V_{ij}=\cup_{b_1,b_2\in\{O,A,B,AB\}}\{V_{b_1,b_2,i,j}\},$$
where $V_{b_1,b_2,i,j}$ is the subset of vertices in $V_{ij}$ with donor blood type $b_1$ and patient blood type $b_2$. For random graph $G$, consider again the optimization problem in \ref{eq::cycle}: 
\begin{equation}
\label{eq::theory-no-fairness}
\max_{x_c\in\{0,1\}}\quad \sum_{c\in \mathcal{C}}x_cu(c) \quad \text{subject to}\quad \sum_{c\in\mathcal{C}}x_c1_{v\in c}\leq 1,\ \forall\ v\in V.
\end{equation}
But here, we allow for a general utility $u$ and a set of cycles or relevant subsets $\mathcal{C}$ with some length limits. The term ``general utility'' refers to a flexible specification of the objective beyond the traditional ``number of transplants'' criterion. Specifically, $u(c)$ can represent any valid utility function that incorporates additional clinical or policy considerations, such as predicted graft survival, donor-recipient age matching, or sensitization adjustments, while still assigning a scalar utility value to each feasible exchange cycle $c$. Regardless of how $u(c)$ is defined, all utilities enter the optimization problem through binary decision variables, so the resulting formulation remains an integer programming problem.

Let $\mathcal{P}$ denote the set of indices $(b_1,b_2,i,j)$, where the subgroup $V_{b_1,b_2,i,j}$ should be prioritized over $V_{b_1,b_2,1-i,j}$. For fixed $\epsilon\geq 0$, we consider the optimization problem under the constraint of prioritizing these subgroups in $\mathcal{P}$: 
\begin{align}
&\max_{x_c\in\{0,1\}}\quad \sum_{c\in \mathcal{C}}x_cu(c)\label{eq::theory-priority-obj}\\
&\begin{aligned}\label{eq::theory-priority-con}
\text{subject to} &\quad \sum_{c\in\mathcal{C}}x_c1_{v\in c}\leq 1,\ \forall\ v\in V,\\
&\sum_{v\in V_{b_1,b_2,i,j}}\sum_{c\in\mathcal{C}}x_c1_{v\in c}-|V_{b_1,b_2,i,j}|\leq \epsilon N\quad \mathrm{or} \quad \sum_{v\in V_{b_1,b_2,1-i,j}}\sum_{c\in\mathcal{C}}x_c1_{v\in c}\leq \epsilon N,\ \forall\ (b_1,b_2,i,j)\in \mathcal{P}.
\end{aligned}
\end{align}
The problems \eqref{eq::theory-no-fairness} and \eqref{eq::theory-priority-obj}--\eqref{eq::theory-priority-con} share the same objective function. The constraint \eqref{eq::theory-priority-con} means that at most $\epsilon N$ in $V_{b_1,b_2,1-i,j}$ can be matched or at most $\epsilon N$ patients in $V_{b_1,b_2,i,j}$ can be unmatched, if the subgroup $V_{b_1,b_2,i,j}$ is prioritized over $V_{b_1,b_2,1-i,j}$. We refer to \eqref{eq::theory-no-fairness} as the unconstrained problem, aside from the vertex-disjointness condition, and to \eqref{eq::theory-priority-obj}--\eqref{eq::theory-priority-con} as the constrained problem.

The following Proposition \ref{prop2} is the main result of this subsection.

\begin{proposition}\label{prop2}
Consider the random graph model in Section \ref{sec::theory}, and fix any sequence $\varepsilon_N \downarrow 0$.
For each $N$, let $\mathrm{Opt}(G_N)$ be the optimum of the unconstrained problem \eqref{eq::theory-no-fairness}, and let $\mathrm{Opt}^{\mathrm{prio}}_{\varepsilon_N}(G_N)$ be the optimum of the constrained problem \eqref{eq::theory-priority-obj}--\eqref{eq::theory-priority-con} with $\epsilon=\epsilon_N$. 
Then, with probability one as $N\to\infty$,
\[
\mathrm{Opt}(G_N) - \mathrm{Opt}^{\mathrm{prio}}_{\varepsilon_N}(G_N) = o(N).
\]
\end{proposition}

Proposition \ref{prop2} shows that we can prioritize patients with either $A=1$ or $A=0$, given their blood type and sensitization level, with ignorable relative efficiency loss when the random graph is large. The result is useful to derive upper bounds for the price of our new fairness criterion, defined as the relative overall utility loss due to the fairness constraint in \eqref{eq::fair-0}. We present these results in the next subsection.

\subsection{Upper bounds for price of our new fairness criterion}\label{subsec::application}

\subsubsection{Optimizing some general utilities}

First, we apply Proposition \ref{prop2} to the scenario with general utilities, allowing for potential recourse strategies. We can obtain a crude upper bound for the price of fairness in Proposition \ref{cor1} below. 

\begin{proposition}\label{cor1}
The price of fairness due to the fairness constraint in \eqref{eq::fair-0} is no greater than
$$\max_{b_1,b_2,r}\max \left\{\frac{\mu_{b_1,b_2,r,1}\overline{\mu}_{r,0}-\mu_{b_1,b_2,r,0}\overline{\mu}_{r,1}}{(\mu_{b_1,b_2,r,1}+\mu_{b_1,b_2,r,0})\overline{\mu}_{r,0}}, \frac{\mu_{b_1,b_2,r,0}\overline{\mu}_{r,1}-\mu_{b_1,b_2,r,1}\overline{\mu}_{r,0}}{(\mu_{b_1,b_2,r,1}+\mu_{b_1,b_2,r,0})\overline{\mu}_{r,1}}\right\}$$
almost surely as $N\to \infty$, where $\overline{\mu}_{r,1}=\sum_{b_1,b_2} \mu_{b_1,b_2,r,1}$ and $\overline{\mu}_{r,0}=\sum_{b_1,b_2} \mu_{b_1,b_2,r,0}$. 
\end{proposition}

If the blood type distributions are balanced across all subgroups defined by different levels of $A$ and $R$, i.e. $\mu_{b_1,b_2,r,0}/\overline{\mu}_{r,0}=\mu_{b_1,b_2,r,1}/\overline{\mu}_{r,1}$ for all $b_1,b_2,r$, then the upper bound in Proposition \ref{cor1} is 0. If the blood type distributions are not balanced within a specific subgroup level $A=a$ and $R=r$, and the optimal solution only matches patients in this subgroup, then the upper bound in Proposition \ref{cor1} is attainable.

\subsubsection{Maximizing the number of transplants without recourse strategies}

Then, we apply Proposition \ref{prop2} to the scenario that maximizes the number of transplants without any recourse strategies. In this scenario, an explicit optimal allocation is explicitly available in \cite{ashlagi2011individual}. 

We introduce the following assumptions to simplify the decomposition of the four-way probability $\mu_{b_1,b_2,r,a}$. First, we assume that the patient and donor in each incompatible pair share the same protected feature level. Second, we assume that the distributions of sensitization levels are consistent across protected feature levels. While these assumptions are not strictly necessary, they facilitate the decomposition of $\mu_{b_1,b_2,r,a}$ into more manageable terms.

Specifically, let $\mu_a$ represent the frequency probability of the protected feature level $A = a$. Define $\mu_{O\mid a}$, $\mu_{A\mid a}$, $\mu_{B\mid a}$, and $\mu_{AB\mid a}$ as the frequency probabilities of blood types O, A, B, and AB, respectively, within the protected feature level $A = a$. Under these assumptions, there exists a constant $c$ such that $\mu_{b_1,b_2,r,a}=cr \mu_a\mu_{b_1\mid a}\mu_{b_2\mid a}$ for all $b_1, b_2 \in \{O, A, B, AB\}$, $r \in \{r_1, \dots, r_M\}$, and $a \in \{0, 1\}$.

Moreover, let $\overline{\mu}_O,\overline{\mu}_A,\overline{\mu}_B,\overline{\mu}_{AB}$ denote the frequency probability of blood types O, A, B, AB, respectively, among the whole population. Let $\overline{\gamma}$ be the average PRA score among the whole population. We can obtain a more precise upper bound for the price of fairness in Proposition \ref{prop3} below. 

\begin{proposition}\label{prop3}

Assume $1.5\overline{\mu}_A>\overline{\mu}_{O}>\overline{\mu}_{A}>\overline{\mu}_{B}>\overline{\mu}_{AB}$ and $\overline{\gamma}<0.4$. Let $\phi_{b_1,b_2}=\sum_{k=0}^1\mu_k\overline{\gamma}\mu_{b_1\mid k}\mu_{b_2\mid k}$, for $b_1,b_2\in\{O,A,B,AB\}$. Let
\begin{equation*}
\begin{split}
T_a(r)=&\mu_ar(\mu_{O\mid a}+\mu_{AB\mid a}-\mu_{O\mid a}\mu_{AB\mid a}+\mu_{A\mid a}^2+\mu_{B\mid a}^2)+2\mu_a\mu_{A\mid a}\mu_{B\mid a},\\
S_a(r)=&T_a(r)+\mu_a\{\mu_{O\mid a}(1-\mu_{O\mid a})+\mu_{A\mid a}\mu_{AB\mid a}+\mu_{B\mid a}\mu_{AB\mid a}\},\\
Q(r)=&T_1(r)+T_0(r)+\phi_{B,AB}+\phi_{O,AB}+\phi_{A,AB}+\phi_{A,O}+\phi_{AB,O}+\phi_{O,B},\\
\end{split}
\end{equation*}
for $a=0,1$ and $r\in\{r_1,\cdots ,r_M\}$, and 
\begin{equation*}
\begin{split}
R_a=&\min\left\{\phi_{B,AB},2\mu_a\mu_{B\mid a}\mu_{AB\mid a}-\phi_{B,AB}\right\}+\min\left\{\phi_{O,AB}+\phi_{A,AB},2\mu_a\mu_{A\mid a}\mu_{AB\mid a}-\phi_{O,AB}-\phi_{A,AB}\right\}\\
&+\min\left\{\phi_{A,O}+\phi_{AB,O},2\mu_a\mu_{O\mid a}\mu_{A\mid a}-\phi_{A,O}-\phi_{AB,O}\right\}+\min\left\{\phi_{O,B},2\mu_a\mu_{O\mid a}\mu_{B\mid a}-\phi_{O,B}\right\},\\
\end{split}
\end{equation*}
for $a=0,1$. Then, the price of fairness due to the fairness constraint in \eqref{eq::fair-0} is no greater than
$$\max_{r\in\{r_1,\cdots ,r_M\}}\max\left\{\frac{S_1T_0-S_0T_1-S_0R_1}{S_1Q}, \frac{S_0T_1-S_1T_0-S_1R_0}{S_0Q},0\right\}$$
almost surely as $N\to\infty$. 
\end{proposition}

Intuitively, the quantities in Proposition \ref{prop3} summarize how much matching potential each protected group has, and how much of that potential can be reallocated when we impose our fairness constraint. For a fixed protected group level \(A=a\) and sensitization level \(R=r\), \(T_a(r)\) is the expected \emph{fraction of incompatible pairs in group \(a\)} (at level \(r\)) that are fully matched in the optimal unconstrained allocation. It aggregates all blood-type combinations that are essentially “always used” in the optimal construction (e.g., X–X pairs, A–B/B–A pairs, and the core part of AB–O, A–O, B–O exchanges) and therefore contribute deterministically to the match count. \(S_a(r)\) is the expected \emph{total fraction of incompatible pairs in group \(a\)} (at level \(r\)) that belong to blood-type classes that can potentially participate in exchanges under the Ashlagi–Roth optimal structure. It includes both the fully matched categories counted in \(T_a(r)\) and the partially matched categories (such as B–AB, A–AB, O–A, O–B). Thus, the ratio \(T_a(r)/S_a(r)\) represents the limiting selection probability for group \(a\) at sensitization level \(r\) in the unconstrained optimum.

At the population level, \(Q(r)\) is the expected \emph{total fraction of incompatible pairs at sensitization level \(r\)} that are matched in the unconstrained optimal allocation, aggregated over both protected groups and all relevant blood-type combinations. It is the denominator in the expression for the asymptotic price of fairness, i.e., the benchmark efficiency against which we compare the fair allocation. Finally, \(R_a\) captures the \emph{additional matching mass that can be shifted} toward group \(A=a\) within the partially matched blood-type classes while preserving feasibility. Each \(\min\{\cdot,\cdot\}\) term in the definition of \(R_a\) corresponds to a specific partially matched blood-type class (B–AB, A–AB, O–A/O–AB, O–B), and ensures we do not reassign more matches to group \(a\) than either (i) the total number of matches available in that class, or (ii) the supply of compatible pairs in group \(a\). In this sense, \(R_a\) measures the “room to maneuver” for reallocating matches across protected groups without changing the blood-type structure of the exchange.

In combination, \(T_a(r)\) and \(S_a(r)\) describe the baseline matching rates for each group and sensitization level, \(Q(r)\) encodes the total unconstrained efficiency, and \(R_a\) quantifies how much we can adjust the allocation in favor of each group when enforcing our fairness constraint. The bound in Proposition \ref{prop3} then compares the best fair allocation achievable under these structural limits with the unconstrained optimum, yielding an explicit asymptotic upper bound on the price of fairness.

The condition $1.5\overline{\mu}_A > \overline{\mu}_O > \overline{\mu}_A > \overline{\mu}_B > \overline{\mu}_{AB}$ imposes a mild and empirically realistic constraint on the population-level blood type distribution, while the condition $\overline{\gamma} < 0.4$ implies that most patients are not highly sensitized. Both assumptions are standard in the kidney exchange literature and appear in \citet{ashlagi2011individual} and \citet{dickerson2014price}. Moreover, these assumptions are consistent with real U.S. blood type frequencies reported by the American Red Cross (2021), approximately 45\% type O, 40\% type A, 11\% type B, and 4\% type AB, and with observed sensitization patterns in national KPD data, where the majority of patients exhibit moderate or low PRA levels. 

As an application of Proposition \ref{prop3}, we present two illustrative examples showing that the asymptotic upper bounds on the price of fairness are very small under empirically supported distributions. In the first example, suppose there are two ethnicity groups, 80\% White American and 20\% African American, with $\mu_{O\mid 1}=0.45$, $\mu_{A\mid 1}=0.40$, $\mu_{B\mid 1}=0.11$, $\mu_{AB\mid 1}=0.04$ and $\mu_{O\mid 0}=0.51$, $\mu_{A\mid 0}=0.26$, $\mu_{B\mid 0}=0.19$, $\mu_{AB\mid 0}=0.04$. Moreover, there are three sensitization levels $\{0.05, 0.45, 0.9\}$. Then, with high probability, the price of fairness converges to zero for any $0.05 < \overline{\gamma} < 0.4$. 

In the second example, suppose there are two ethnicity groups, 90\% White American and 10\% Asian American, with $\mu_{O\mid 1}=0.45$, $\mu_{A\mid 1}=0.40$, $\mu_{B\mid 1}=0.11$, $\mu_{AB\mid 1}=0.04$ and $\mu_{O\mid 0}=0.40$, $\mu_{A\mid 0}=0.275$, $\mu_{B\mid 0}=0.255$, $\mu_{AB\mid 0}=0.07$. Moreover, there are five sensitization levels $\{0.05, 0.25, 0.45, 0.65, 0.9\}$. Then, with high probability, the price of fairness is less than $0.01$ for any $0.05 < \overline{\gamma} < 0.09$ and converges to zero for any $0.09 \leq \overline{\gamma} < 0.40$. 

These two examples demonstrate that, under empirically grounded population characteristics, the price of fairness in large kidney exchange graphs is negligible, typically below 1\%, confirming that fairness constraints can be imposed with minimal efficiency loss in real-world settings.

Under a similar random graph model, \cite{dickerson2014price} claim that the price of group fairness is no greater than $2/33$, as $n\to \infty$; \cite{ashlagi2011individual} claim that the relative efficiency loss for individual rationality is only about $1\%$ in multi-hospital kidney exchange. All these results give very low efficiency loss because, in large random graph models, there is a rich set of edges from each vertex such that we can easily adjust for the optimal solution such that the fairness constraint is satisfied. Beyond the kidney exchange setting, \cite{bertsimas2011price} provide an upper bound for the price of proportional fairness and the price of max-min fairness, which are close to 1, in general allocation problems.

\section{Numerical studies}\label{sec::numerical}

In this section, we present simulation studies based on the random graph models described in Section \ref{sec::theory} and real data from the UNOS dataset. All the optimization problems are solved by Gurobi (version 11.0) in R. 

\subsection{Simulation under the random graph models}\label{subsec::random_graph}

We first conduct a simulation study under the random graph models to evaluate the numerical performance of different fairness criteria. We consider a binary protected feature that indicates whether one is white or non-white. Specifically, we fix 80 white incompatible pairs (80\%) and 20 non-white incompatible pairs (20\%). Among the 80 white pairs, 56 patients (70\%) are lowly-sensitized, 16 patients (20\%) are moderately-sensitized, and 8 patients (10\%) are highly-sensitized; and among the 20 non-white pairs, 14 patients (70\%) are lowly-sensitized, 4 patients (20\%) are moderately-sensitized, and 2 patients (10\%) are highly-sensitized. The PRA scores are set to be 0.9, 0.45, and 0.05 for highly-sensitized, moderately-sensitized, and lowly-sensitized patients, respectively. The distribution follows the analysis in \cite{saidman2006increasing}. 

For non-white donors and patients, we simulate the blood type from a multinomial distribution (O: 51\%, A: 26\%, B: 19\%, AB: 4\%). For white donors and patients, we simulate the blood type from another multinomial distribution (O: 45\%, A: 40\%, B: 11\%, AB: 4\%). If a donor and a patient are blood-type compatible, they are incompatible with the probability of the patient's PRA score; otherwise, they are incompatible with probability 1. Given the patients, we repeat sampling the donors and dropping compatible pairs until the numbers of incompatible pairs are reached.  Similarly, based on the blood types and PRA scores, we simulate the edges of the graph. That is, we simulate the compatibility for the donor and patient of every two vertices in the graph. We only allow cycles of length at most 3, and we do not consider any recourse strategies in this section. 

We compare three fairness criteria, group fairness, individual fairness, and our newly proposed fairness criterion, using the outcome with no fairness constraints as a baseline. For our new fairness criterion, we evaluate the two parameter choices introduced in Section \ref{subsec::fairness_definition}. For group fairness, we consider two parameter-selection methods. The first selects the largest feasible value of $\alpha$ for which \eqref{eq::group} remains solvable, following \cite{dickerson2014price}. This corresponds to imposing the strongest possible constraint, requiring us to maximize the number of matched highly sensitized patients. The second method selects the largest $\alpha$ that preserves the overall utility achieved in \eqref{eq::cycle}, as in \cite{freedman2020adapting}. In other words, among all exchange plans that maximize the objective in \eqref{eq::group}, we choose the one that yields the greatest number of matched highly sensitized patients. For individual fairness, we adopt two analogous parameter-selection approaches. Specifically, we minimize the $L_1$ norm of the difference vector $\pi - \overline{\pi}$ while retaining at least 80\% and 100\% of the maximum achievable utility, respectively. These correspond to a stronger and a weaker fairness constraint.

Figure \ref{fig1} presents the average selection rates of different fairness criteria within each subgroup. From Figure \ref{fig1}, group fairness works to increase the selection rates of highly-sensitized patients, and individual fairness works to balance the selection rates of the six subgroups. Differently, our fairness balances the selection rates of the white and non-white patients within each sensitization stratum, instead of the selection rates of all six subgroups. Our stronger fairness constraint can almost eliminate the absolute difference in selection rates, while our weaker fairness constraint can also significantly reduce the absolute difference in selection rates compared to other methods. For the price of fairness, the relative utility loss is around 6.3\% for our stronger fairness constraint and 0.7\% for our weaker fairness constraint.  

\begin{figure}[h]
\caption{Simulation results under random graph models. The average selection rates within each subgroup are calculated over 100 data replications. The error bars represent the mean $\pm$ 1 standard deviation of the absolute differences in selection rates. }\label{fig1}
\centering
\includegraphics[width=0.8\textwidth]{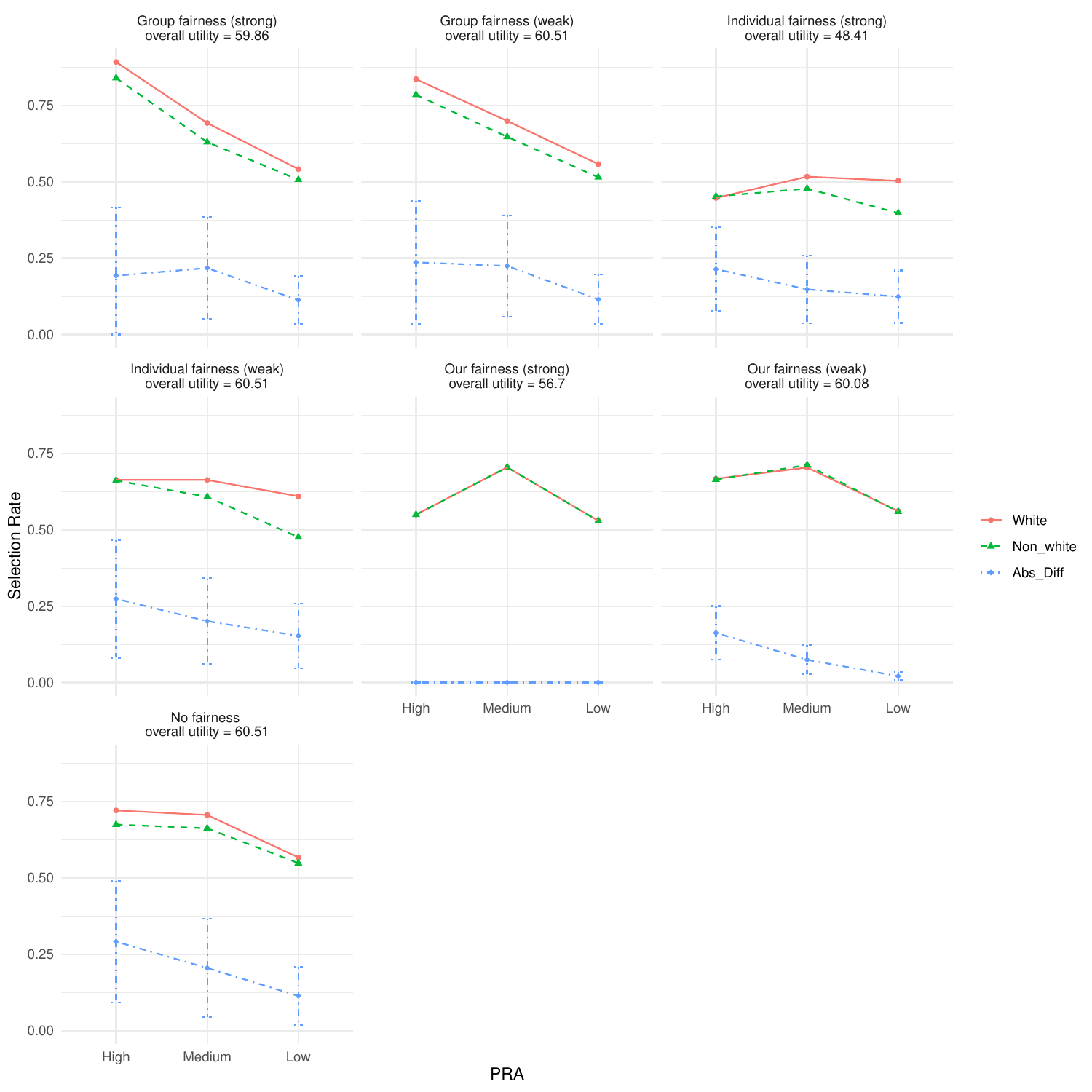}
\end{figure}

\subsection{UNOS data analysis}\label{subsec::unos}

We next conduct a simulation study based on the National UNOS STAR dataset. The National UNOS STAR dataset provides comprehensive transplant records collected by the UNOS, covering donor and recipient characteristics, allocation details, and transplant outcomes. After removing the missing values, the dataset comprises 77,073 records of transplant information for donors and patients, which include details such as blood types, HLA antigens (A1, A2, B1, B2, DR1, DR2), PRA scores, and racial background. The protected feature is set to be a binary variable: white (64.8\%) and other racial backgrounds (35.2\%). Patients are categorized based on their PRA scores as follows: those with scores above 0.8 are labeled as highly sensitized; those with scores ranging from 0.1 to 0.8 are considered moderately sensitized; and patients with scores below 0.1 are labeled as lowly sensitized. 

We fix the number of incompatible donor-patient pairs to be 100. Donors and patients are randomly sampled from the dataset to form incompatible pairs independently. The overall compatibility is determined by both blood type and HLA compatibility. HLA compatibility is assessed based on the number of mismatches in the A, B, and DR alleles. Specifically, a donor and patient are considered HLA compatible if their level of HLA mismatch is less than 3. The HLA mismatch level is based on UK kidney matching policies, and it can be calculated using the R package \texttt{transplantr}. Again, we only allow cycles of length at most 3. 

Figure \ref{fig3} reports the average selection rates within each subgroup. From Figure \ref{fig3}, group fairness, individual fairness, and our new fairness criterion all enhance equitable access to transplant opportunities according to their specific fairness criteria.

\begin{figure}[h]
\caption{Simulation results based on UNOS data. The average selection rates within each subgroup are calculated over 100 data replications. The error bars represent the mean $\pm$ 1 standard deviation of the absolute differences in selection rates.  }\label{fig3}
\centering
\includegraphics[width=0.8\textwidth]{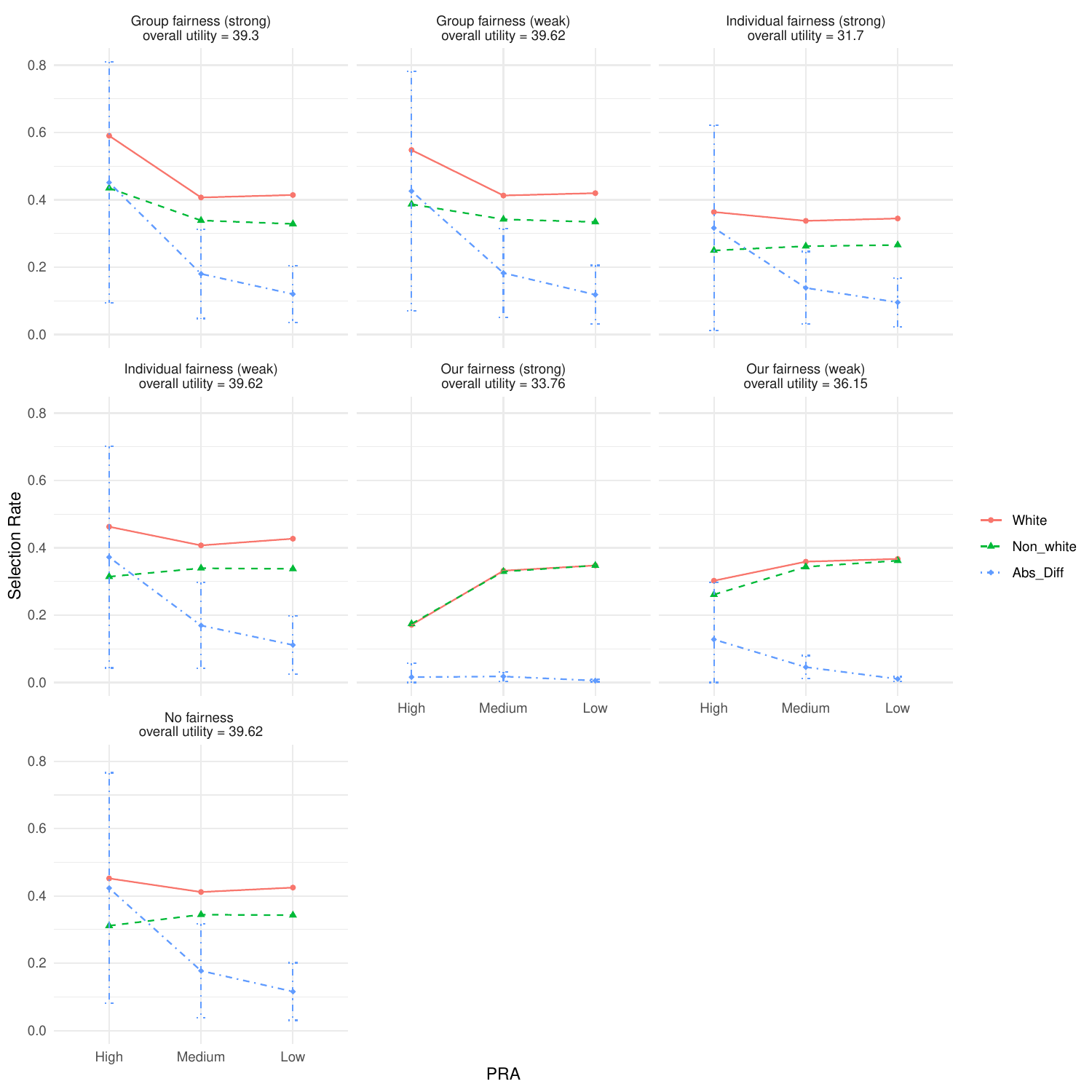}
\end{figure}

\subsection{Experiment on selection probability prediction}\label{subsec::numerical}

We next conduct a numerical experiment to evaluate the prediction accuracy of the selection probability. We assume there is an underlying population consisting of 80\% white and 20\% non-white incompatible pairs, following the same blood type distributions and sensitization distributions in Section \ref{subsec::random_graph}. We assume the historical data $\{\tilde{v}_1,\cdots ,\tilde{v}_{N_0}\}$, current data $\{v_1,\cdots ,v_{N_1}\}$, and future data are independently sampled from the population. We fix $N_0=200$, $N=100$, $B=1000$, and vary $N_1$ in $\{20,40,60,80\}$. For illustration purposes, we focus on our new fairness criterion.

Table \ref{tab1} presents the numerical results of the selection probability prediction. As $L$ increases, the prediction accuracy improves, evidenced by a lower mean squared error (MSE) and a coverage rate of the prediction interval approaching 95\%. This improvement occurs because the size of the unobserved future data, $N-N_1$, decreases, which makes the prediction easier. Although these results indicate the method's validity, there is a bias in the prediction interval due to the finite size of $N_0$. When $N_0$ is large, the set $\{\tilde{v}_1,\cdots ,\tilde{v}_{N_0}\}$ closely approximates the underlying population distribution; while when $N_0$ is small, the distribution of $\{\tilde{v}_1,\cdots ,\tilde{v}_{N_0}\}$ may differ from the underlying population distribution. 

\begin{table}
\centering
\begin{tabular}{|l|c|c|c|c|}
\hline
$N_1$ & 20 & 40 & 60 & 80 \\
\hline
MSE & 0.079 & 0.056 & 0.045 & 0.038 \\
\hline
Coverage & 0.977 & 0.965 & 0.962 &  0.958 \\
\hline
Width & 0.445 & 0.332 & 0.265 & 0.227 \\
\hline
\end{tabular}
\caption{Accuracy of the predicted selection probability of $v_1,\cdots ,v_{N_1}$. The results are averaged over 100 data replications.}\label{tab1}
\end{table}

\section{Discussion}\label{sec::conclusion}

In this paper, we propose a new fairness criterion that balances selection probabilities within protected groups across each sensitization level. Based on the calibration principle in machine learning, this fairness criterion offers a meaningful and innovative approach in the context of kidney exchange. We propose an efficient solution to implement this criterion and conduct both theoretical and empirical evaluations to analyze the associated price of fairness.

\bibliographystyle{chicago}      
\bibliography{fairness}   

\newpage

\appendix

\renewcommand{\thetheorem}{S\arabic{theorem}}
\setcounter{theorem}{0}
\renewcommand{\thelemma}{S\arabic{lemma}}
\setcounter{lemma}{0}
\renewcommand{\theproposition}{S\arabic{proposition}}
\setcounter{proposition}{0}
\renewcommand{\thecorollary}{S\arabic{corollary}}
\setcounter{corollary}{0}
\renewcommand{\thedefinition}{S\arabic{definition}}
\setcounter{definition}{0}
\renewcommand{\thealgo}{S\arabic{algo}}
\setcounter{algo}{0}
\renewcommand{\thepage}{S\arabic{page}}
\setcounter{page}{1}
\renewcommand{\theequation}{S\arabic{equation}}
\setcounter{equation}{0}
\renewcommand{\thetable}{S\arabic{table}}
\setcounter{table}{0}
\renewcommand{\thefigure}{S\arabic{figure}}
\setcounter{figure}{0}

\begin{center}
\huge \bfseries  Supplementary material
\end{center}

\bigskip 

Section \ref{sec::appendix-proofs} presents all the mathematical proofs. Section \ref{sec::appendix-simulation} presents more simulations
with failure-ware strategies.

\section{Proofs}\label{sec::appendix-proofs}

The proof of Proposition \ref{prop2} depends on the following lemma. 

\begin{lemma}\label{lem2}
Fix $k \ge 2$ and $p \in (0,1)$. For each $n$, let $(A_1,\dots,A_k)$ be a $k$-partite undirected
random graph with $|A_i| = n$ for all $i$, where every potential edge between distinct parts is present independently with probability at least $p$.  

Then, with probability $1 - o(1)$ as $n \to \infty$, the graph contains $n$ vertex-disjoint copies of the complete $k$-partite graph with one vertex in each part; that is, there exist $n$ disjoint $k$-tuples, each consisting of one vertex from every part $A_i$, such that all $\binom{k}{2}$ inter-part edges among the vertices in each $k$-tuple are present.
\end{lemma}

\begin{proof}[Proof of Lemma \ref{lem2}]
We build the statement in four steps.

\medskip
\noindent\textbf{Step 1: Bipartite building block.}
Let $X,Y$ be disjoint sets of size $n$, and consider a bipartite graph in which each edge $(x,y)$ is present independently with probability at least $q>0$.  By monotone coupling with $G(n,n,q)$, this random graph stochastically dominates an Erd\H{o}s--R\'enyi bipartite graph.  It is classical that $G(n,n,q)$ has a perfect matching with probability $1-o(1)$ as $n\to\infty$, and “having a perfect matching’’ is monotone in the edge set.  Hence the original bipartite graph also has a perfect matching with probability $1-o(1)$.

\medskip
\noindent\textbf{Step 2: Uniform representation and filtration.}
For each unordered pair $(u,v)$ with $u\in A_i$, $v\in A_j$, $i\neq j$, let $U_{(u,v)}$ be independent $\mathrm{Unif}[0,1]$.  If the designated probability of the edge $(u,v)$ is $p_{(u,v)}\in[p,1]$, we declare $(u,v)$ present iff $U_{(u,v)}\le p_{(u,v)}$.  This yields exactly the model in the lemma.

For $t=1,\dots,k-1$, define
\[
\mathcal{F}_t
:= \sigma\bigl( U_{(u,v)} : u\in A_i,\ v\in A_j,\ i\neq j,\ \{u,v\}\subseteq A_1\cup\cdots\cup A_{t+1}\bigr),
\]
which reveals all edges whose endpoints lie in $A_1,\dots,A_{t+1}$.

\medskip
\noindent\textbf{Step 3: Inductive construction of $(t+1)$-tuples.}
We build the desired $k$-tuples in $k-1$ rounds.  At the end of round $t\ge2$ we will have $n$ disjoint sets $B^{(t)}_1,\dots,B^{(t)}_n$, each containing one vertex from each of $A_1,\dots,A_t$ and inducing a complete $t$-partite subgraph.

\emph{Round $t=2$.}  
Consider the bipartite graph between $A_1$ and $A_2$.  All edges are independent and present with probability at least $p$, so by Step 1 it has a perfect matching with probability $1-o(1)$.  Fix such a matching and write its edges as $(x_j,y_j)$, $j=1,\dots,n$, with $x_j\in A_1$, $y_j\in A_2$, and define
\[
B^{(2)}_j:=\{x_j,y_j\},\qquad B_2:=\{B^{(2)}_1,\dots,B^{(2)}_n\}.
\]

\emph{Inductive step.}  
Assume for some $t\in\{2,\dots,k-1\}$ that we have disjoint sets
\[
B_t = \{B^{(t)}_1,\dots,B^{(t)}_n\}
\]
such that each $B^{(t)}_j$ contains exactly one vertex from each of $A_1,\dots,A_t$ and induces a complete $t$-partite subgraph.  Then $B_t$ is $\mathcal{F}_{t-1}$-measurable.

Define a bipartite graph $\Gamma_t$ whose left part is $B_t$ and whose right part is $A_{t+1}$.  For $B\in B_t$ and $a\in A_{t+1}$, put an edge between $B$ and $a$ in $\Gamma_t$ if and only if $a$ is adjacent to every vertex in $B$ in the original graph, i.e.
\[
(B,a)\in E(\Gamma_t)
\quad\Longleftrightarrow\quad
U_{(v,a)}\le p_{(v,a)}\ \text{for all } v\in B.
\]
Conditional on $\mathcal{F}_{t-1}$, $B_t$ is fixed and the variables $\{U_{(v,a)}:v\in B\}$ are independent, with $|B|=t$ and $p_{(v,a)}\ge p$.  Hence
\[
\mathrm{Pr}\bigl((B,a)\in E(\Gamma_t)\mid\mathcal{F}_{t-1}\bigr)
=\prod_{v\in B} p_{(v,a)}
\ge p^t=:q_t>0.
\]
For distinct pairs $(B,a)\neq (B',a')$, the sets of uniforms $\{U_{(v,a)}:v\in B\}$ and $\{U_{(v',a')}:v'\in B'\}$ are disjoint, so the edge indicators in $\Gamma_t$ are independent, each with conditional probability at least $q_t$.

Thus, conditional on $\mathcal{F}_{t-1}$, $\Gamma_t$ satisfies the assumptions of Step 1 with parameter $q_t$, and since $|B_t|=|A_{t+1}|=n$, it has a perfect matching with probability $1-o(1)$.  On that event, write the matching as $(B^{(t)}_j,a_j)$, $j=1,\dots,n$, and set
\[
B^{(t+1)}_j := B^{(t)}_j\cup\{a_j\},\qquad
B_{t+1} := \{B^{(t+1)}_1,\dots,B^{(t+1)}_n\}.
\]
Each $B^{(t+1)}_j$ is then a complete $(t+1)$-partite subgraph with one vertex in each of $A_1,\dots,A_{t+1}$.

\medskip
\noindent\textbf{Step 4: Union bound over rounds.}
There are $k-1$ rounds.  For each fixed $t$, the conditional failure probability is $o(1)$ as $n\to\infty$, hence the unconditional failure probability is $o(1)$.  Since $k$ is fixed, a union bound shows that with probability $1-o(1)$ all rounds succeed, yielding $n$ disjoint $k$-tuples as in the statement.
\end{proof}

\begin{proof}[Proof of Proposition \ref{prop2}]
We organize the proof into six steps:
\begin{itemize}
\item Step 1 defines coarse and fine types and their frequencies.
\item Step 2 introduces template types and a coarse-type linear program.
\item Step 3 shows that the unconstrained optimal value is asymptotic to the LP optimum.
\item Step 4 splits coarse usage into fine types in a way that encodes priorities.
\item Step 5 realizes these fine-type counts in $G_N$ using Lemma \ref{lem2}.
\item Step 6 combines everything.
\end{itemize}

Throughout, we work under the assumptions of Section \ref{sec::theory}.

\medskip
\noindent\textbf{Step 1: Coarse and fine types and their frequencies.}
Each vertex $v$ (an incompatible donor--patient pair) has
\[
b_1(v), b_2(v)\in\{O,A,B,AB\},\quad r(v)\in\{r_1,\dots,r_M\},\quad A(v)\in\{0,1\}.
\]
Define the coarse type
\[
\tilde\tau(v)
:= (b_1(v),b_2(v),r(v))
\in\widetilde{\mathbb{T}}
:=\{O,A,B,AB\}^2\times\{r_1,\dots,r_M\},
\]
and the fine type
\[
\tau(v)
:= (\tilde\tau(v),A(v))
\in\mathbb{T}
:=\widetilde{\mathbb{T}}\times\{0,1\}.
\]

For each $N$, let $V_N$ denote the vertex set of $G_N$, with $|V_N|=N$.  For $t\in\mathbb{T}$ and $\tilde t\in\widetilde{\mathbb{T}}$, define
\[
n_t(N) := \bigl|\{v\in V_N:\tau(v)=t\}\bigr|,
\qquad
n_{\tilde t}(N) := \sum_{a=0}^1 n_{(\tilde t,a)}(N),
\]
and let $\pi_t>0,\pi_{\tilde t}>0$ be the corresponding type probabilities under the i.i.d.\ type distribution.

By the strong law of large numbers, almost surely
\begin{equation}\label{eq:LLN-types-prop2-again}
\frac{n_t(N)}{N} \to \pi_t,
\qquad
\frac{n_{\tilde t}(N)}{N} \to \pi_{\tilde t}
\quad\text{as } N\to\infty\ \text{for all }t,\tilde t.
\end{equation}
We henceforth work on any realization for which \eqref{eq:LLN-types-prop2-again} holds.

\medskip
\noindent\textbf{Step 2: Template types and a coarse-type linear program.}
Ignoring the protected feature, we classify cycles (or relevant subsets) into finitely many “templates”.  Two sets $c_1,c_2\in\mathcal{C}(G_N)$ have the same \emph{template type} if there is a bijection $\phi:V(c_1)\to V(c_2)$ such that:
\begin{itemize}
  \item $\tilde\tau(v)=\tilde\tau(\phi(v))$ for all $v\in V(c_1)$; and
  \item $(v_1,v_2)\in E_N$ iff $(\phi(v_1),\phi(v_2))\in E_N$.
\end{itemize}
Because the size $|c|$ is uniformly bounded and $\widetilde{\mathbb{T}}$ is finite, there are only finitely many template types.  Index them by $i=1,\dots,K$.

For each template $i$, fix a representative directed pattern $H_i$ with vertex set $V(H_i)$ and $|V(H_i)|=s_i$.  For each coarse type $\tilde t\in\widetilde{\mathbb{T}}$, let
\[
\tilde a_{\tilde t,i}
:= \bigl|\{w\in V(H_i):\tilde\tau(w)=\tilde t\}\bigr|.
\]
Since $u(c)$ depends only on the induced subgraph and the coarse types (and is invariant under relabelings preserving this information), $u_i:=u(c)$ is well-defined for each template $i$.

We consider the coarse-type linear program
\begin{equation}\label{eq:LP-coarse-prop2-again}
\begin{aligned}
\max_{z\geq 0}\quad & U^\T z := \sum_{i=1}^K u_i z_i, \\
\text{s.t.}\quad
& \sum_{i=1}^K \tilde a_{\tilde t,i} z_i \le \pi_{\tilde t}, \quad \forall \tilde t\in\widetilde{\mathbb{T}}.
\end{aligned}
\end{equation}
Feasibility is clear and boundedness follows from $u(c)/|c|\le C$ and bounded template sizes.  Let $\rho$ be the optimal value and $z^\star$ an optimal solution.

\medskip
\noindent\textbf{Step 3: Asymptotic behavior of the unconstrained optimum.}
We show that
\begin{equation}\label{eq:Opt-asymptotic-prop2-again}
\mathrm{Opt}(G_N)=\rho N+o(N)
\quad\text{almost surely.}
\end{equation}

\smallskip
\emph{Step 3A: Upper bound.}
Let $(x_c)_{c\in\mathcal{C}(G_N)}$ be an optimal solution of the unconstrained problem \eqref{eq::theory-no-fairness} on $G_N$.  
For each template type $i$, define
\[
Z^{(N)}_i
:= \sum_{c\in\mathcal{C}(G_N)} x_c\,1_{\{c\text{ has template type }i\}},
\qquad
y^{(N)} := \frac{1}{N}Z^{(N)}.
\]
Vertex-disjointness implies that, for each coarse type $\tilde t$,
\[
\sum_{i=1}^K \tilde a_{\tilde t,i} Z^{(N)}_i
\le n_{\tilde t}(N).
\]
Dividing by $N$ and letting $N\to\infty$, any subsequential limit $y$ of $\{y^{(N)}\}$ satisfies
\[
\sum_{i=1}^K \tilde a_{\tilde t,i} y_i \le \pi_{\tilde t},
\qquad\forall\tilde t\in\widetilde{\mathbb{T}},
\]
so $y$ is feasible for \eqref{eq:LP-coarse-prop2-again} and $U^\T y\le\rho$.

The objective value on $G_N$ is
\[
\mathrm{Opt}(G_N) = \sum_{i=1}^K u_i Z^{(N)}_i = N U^\T y^{(N)},
\]
so
\[
\limsup_{N\to\infty}\frac{\mathrm{Opt}(G_N)}{N}
= \limsup_{N\to\infty}U^\T y^{(N)}
\le \rho.
\]

\smallskip
\emph{Step 3B: Lower bound.}
Define integer vectors
\[
Z^{(N)} := \bigl(\lfloor Nz^\star_1\rfloor,\dots,\lfloor Nz^\star_K\rfloor\bigr),
\qquad
y^{(N)} := \frac{1}{N}Z^{(N)}.
\]
Then $y^{(N)}\to z^\star$, and for each coarse type $\tilde t$,
\[
\sum_{i=1}^K \tilde a_{\tilde t,i} y^{(N)}_i
\le \sum_{i=1}^K \tilde a_{\tilde t,i} z^\star_i
\le \pi_{\tilde t},
\]
so $\sum_{i} \tilde a_{\tilde t,i} Z_i^{(N)} \le N\pi_{\tilde t}+O(1)$.  Together with \eqref{eq:LLN-types-prop2-again}, this implies $\sum_i \tilde a_{\tilde t,i} Z_i^{(N)} \le n_{\tilde t}(N)$ for all large $N$, so the coarse-type supply of vertices suffices to realize $Z^{(N)}$ templates.

To realize these templates, fix $i$.  Label the vertices of $H_i$ as $w_1,\dots,w_{s_i}$ with coarse types $\tilde t_1,\dots,\tilde t_{s_i}$.  For each $(i,\ell)$ choose a subset
\[
A^{(i)}_\ell(N)\subseteq\{v\in V_N : \tilde\tau(v)=\tilde t_\ell\}
\]
of size $Z^{(N)}_i$, such that all sets $A^{(i)}_\ell(N)$ are disjoint across all $(i,\ell)$; this is possible by the previous capacity inequality.  

For fixed $i$, consider the induced $s_i$-partite undirected graph on parts $A^{(i)}_1(N),\dots,A^{(i)}_{s_i}(N)$.  
Let $L_i$ be the number of directed edges in $H_i$.  
If $u\in A^{(i)}_\ell(N)$ and $v\in A^{(i)}_m(N)$ with $\ell\neq m$, we connect $u$ and $v$ by an undirected edge if all directed edges between $w_\ell$ and $w_m$ required by $H_i$ (and medically admissible) are present between $u$ and $v$ in $G_N$.  Each such undirected edge appears with probability at least $p_0^{L_i}=:p_i>0$, and different pairs $(u,v)$ use disjoint sets of underlying directed edges, so these edges are independent.

Thus, for each $i$, the $s_i$-partite graph on $A^{(i)}_1(N),\dots,A^{(i)}_{s_i}(N)$ satisfies the hypotheses of Lemma \ref{lem2} with $k=s_i$, $n=Z^{(N)}_i$, and parameter $p_i>0$.  If $z^\star_i>0$, then $Z^{(N)}_i\to\infty$ and Lemma \ref{lem2} implies that this graph contains $Z^{(N)}_i$ vertex-disjoint complete $s_i$-partite subgraphs with probability $1-o(1)$.  Each such subgraph corresponds to one copy of $H_i$.  If $z^\star_i=0$, then $Z^{(N)}_i=O(1)$ and failure to realize these copies changes the total utility by at most $O(1)=o(N)$.

Since there are finitely many templates, a union bound shows that with probability $1-o(1)$ we can realize all templates $H_i$ with counts $Z^{(N)}_i$ simultaneously.  On this event, we obtain an exchange plan with template counts $Z^{(N)}$ and utility
\[
\sum_{i=1}^K Z^{(N)}_i u_i
= N\,U^\T z^\star + O(1)
= \rho N + o(N).
\]
Hence $\liminf_{N\to\infty}\frac{\mathrm{Opt}(G_N)}{N}\ge\rho$, proving \eqref{eq:Opt-asymptotic-prop2-again}.

\medskip
\noindent\textbf{Step 4: Splitting coarse usage into fine types and encoding priority.}
We now refine coarse-type usage to fine types (including the protected feature) in a way that captures the desired subgroup prioritization.

\smallskip
\emph{Step 4A: Coarse usage per type.}
For each coarse type $\tilde t$, define
\[
\lambda_{\tilde t}
:= \sum_{i=1}^K \tilde a_{\tilde t,i} z^\star_i
\le \pi_{\tilde t}.
\]

\smallskip
\emph{Step 4B: Splitting into fine-type usages.}
Choose nonnegative numbers $\lambda_t$ for $t=(\tilde t,a)\in\mathbb{T}$ such that
\[
\lambda_{(\tilde t,0)}+\lambda_{(\tilde t,1)} = \lambda_{\tilde t},
\qquad
\lambda_t\le \pi_t\quad\forall t.
\]
For each prioritized quadruple $(b_1,b_2,i,j)\in\mathcal{P}$ with coarse type $\tilde t=(b_1,b_2,r_j)$ and fine types
\[
t^{(\mathrm{prio})}=(\tilde t,i), \qquad t^{(\mathrm{unprio})}=(\tilde t,1-i),
\]
we further require that either $\lambda_{t^{(\mathrm{unprio})}}=0$ or $\lambda_{t^{(\mathrm{prio})}}=\pi_{t^{(\mathrm{prio})}}$.
This is achieved by the rule
\[
(\lambda_{t^{(\mathrm{prio})}},\lambda_{t^{(\mathrm{unprio})}})
=
\begin{cases}
(\lambda_{\tilde t},0), & \lambda_{\tilde t}\le\pi_{t^{(\mathrm{prio})}},\\[0.2em]
(\pi_{t^{(\mathrm{prio})}},\lambda_{\tilde t}-\pi_{t^{(\mathrm{prio})}}), & \lambda_{\tilde t}>\pi_{t^{(\mathrm{prio})}}.
\end{cases}
\]
For coarse types not in any prioritized pair, we pick any feasible split.

\smallskip
\emph{Step 4C: Fine-type template counts.}
We refine the coarse counts $\tilde a_{\tilde t,i}$ into $b_{t,i}\ge0$ so that
\begin{equation}\label{eq:split-constraints-prop2-again}
\sum_{a=0}^1 b_{(\tilde t,a),i} = \tilde a_{\tilde t,i}\quad\forall \tilde t,i,
\qquad
\sum_{i=1}^K b_{t,i} z^\star_i = \lambda_t\quad\forall t\in\mathbb{T}.
\end{equation}
If $\lambda_{\tilde t}>0$, we set
\[
b_{(\tilde t,a),i}
:= \tilde a_{\tilde t,i}\,\frac{\lambda_{(\tilde t,a)}}{\lambda_{\tilde t}},
\quad a\in\{0,1\},
\]
and if $\lambda_{\tilde t}=0$, we set $b_{(\tilde t,a),i}=0$.  A direct check shows that \eqref{eq:split-constraints-prop2-again} holds.

\medskip
\noindent\textbf{Step 5: Constructing matching fine-type counts in $G_N$.}
We now pass from the limiting fractions $\lambda_t$ and $b_{t,i}$ to integer counts in $G_N$.

\smallskip
\emph{Step 5A: Target number of matched vertices per fine type.}
For each fine type $t$ and $N$, let
\[
M^{(N)}_t := \bigl\lfloor N\lambda_t\bigr\rfloor.
\]
By \eqref{eq:LLN-types-prop2-again} and $\lambda_t\le\pi_t$, we have $M^{(N)}_t\le n_t(N)$ for all large $N$, so we can choose
\[
V^{\mathrm{use}}_t(N)\subseteq\{v\in V_N:\tau(v)=t\},
\qquad|V^{\mathrm{use}}_t(N)|=M^{(N)}_t.
\]
For each prioritized quadruple $(b_1,b_2,i,j)\in\mathcal{P}$ with fine types $t^{(\mathrm{prio})},t^{(\mathrm{unprio})}$, our construction of $\lambda_t$ implies that, up to $o(N)$ vertices, either only the prioritized type is used or the prioritized type is essentially saturated.  This exactly matches the asymptotic priority requirement in the definition of $\mathrm{Opt}^{\mathrm{prio}}_{\varepsilon_N}(G_N)$ for some sequence $\varepsilon_N\downarrow0$.

\smallskip
\emph{Step 5B: Matching fine-type/template counts and realization.}
Let $\gamma_{t,i} := b_{t,i} z^\star_i$.  We seek integers $m^{(N)}_{t,i}$ such that
\[
\sum_{i=1}^K m^{(N)}_{t,i} = M^{(N)}_t \quad\forall t\in\mathbb{T},
\qquad
\sum_{t\in\mathbb{T}} m^{(N)}_{t,i} = s_i Z^{(N)}_i\quad\forall i,
\]
and $m^{(N)}_{t,i} = N\gamma_{t,i} + O(1)$ as $N\to\infty$.  This can be obtained by a standard rounding argument on the finite transportation polytope: start from $m^{(N)}_{t,i}:=\lfloor N\gamma_{t,i}\rfloor$ and adjust finitely many entries by $O(1)$ to match the row and column sums exactly.  The total adjustment per $(t,i)$ is $O(1)$, so this changes the utility by at most $O(1)=o(N)$.

For each pair $(t,i)$, partition $V^{\mathrm{use}}_t(N)$ into disjoint subsets of size $m^{(N)}_{t,i}$ and assign each subset to the appropriate position in the $Z^{(N)}_i$ copies of $H_i$.  This yields, for each template $i$, parts $A^{(i)}_1(N),\dots,A^{(i)}_{s_i}(N)$ of size $Z^{(N)}_i$ with prescribed fine (and hence coarse) types.

We now apply exactly the same Lemma \ref{lem2}-based construction as in Step 3B, restricted to these parts.  The independence and lower bounds on edge probabilities are unchanged, so with probability $1-o(1)$ we can realize the desired $Z^{(N)}_i$ disjoint copies of $H_i$ using precisely these vertices, for all $i$ with $z^\star_i>0$.  Templates with $z^\star_i=0$ contribute $O(1)$ vertices and can be handled arbitrarily, with $o(N)$ effect.

Thus, with probability $1-o(1)$, for all large $N$ there exists an exchange plan that:
\begin{itemize}
  \item uses exactly the vertices in $\bigcup_t V^{\mathrm{use}}_t(N)$ up to $O(1)$,
  \item has template counts $Z^{(N)}$, and
  \item satisfies the asymptotic priority requirement with sequence $\varepsilon_N$.
\end{itemize}
The corresponding utility is
\[
\sum_{i=1}^K Z^{(N)}_i u_i + O(1)
= \rho N + o(N),
\]
so almost surely
\[
\mathrm{Opt}^{\mathrm{prio}}_{\varepsilon_N}(G_N)
\ge \rho N + o(N).
\]

\medskip
\noindent\textbf{Step 6: Conclusion.}
From Step 3 we know $\mathrm{Opt}(G_N)=\rho N+o(N)$ almost surely, and from Step 5 we know $\mathrm{Opt}^{\mathrm{prio}}_{\varepsilon_N}(G_N)\ge \rho N+o(N)$ almost surely.  Hence
\[
\mathrm{Opt}(G_N) - \mathrm{Opt}^{\mathrm{prio}}_{\varepsilon_N}(G_N)
= o(N)
\quad\text{almost surely,}
\]
proving the proposition.
\end{proof}

\begin{proof}[Proof of Proposition \ref{cor1}]
We bound the price of fairness in terms of the number of matched vertices and show it is at most the explicit quantity $L$ in the statement.  The argument is purely type-level and ignores graph structure until the final step.

\medskip
\noindent\textbf{Step 1: Local loss for a fixed \((b_1,b_2,r)\).}
Fix $(b_1,b_2)$ and $r$, and consider the two fine types
\[
(b_1,b_2,r,1)\quad\text{and}\quad(b_1,b_2,r,0),
\]
with proportions $\mu_{b_1,b_2,r,1}$ and $\mu_{b_1,b_2,r,0}$.  Let
\[
\overline{\mu}_{r,a}
:= \sum_{b_1,b_2}\mu_{b_1,b_2,r,a},\qquad a\in\{0,1\},
\]
be the total mass at level $r$ in group $A=a$.

In the unconstrained optimum, the total fraction of matched vertices of these two types can be written as
\[
p\,(\mu_{b_1,b_2,r,1}+\mu_{b_1,b_2,r,0})
\]
for some $p\in[0,1]$, which we can think of as the local selection probability at $(b_1,b_2,r)$.  We now compare this with an allocation that enforces the stronger requirement that the selection probabilities at level $r$ for $A=0$ and $A=1$ are exactly equal:
\[
\frac{\text{matches in }(A=1,R=r)}{\overline{\mu}_{r,1}}
=
\frac{\text{matches in }(A=0,R=r)}{\overline{\mu}_{r,0}}.
\]
This is stricter than the original fairness constraint, so any loss bound obtained here is conservative.

If the ratio $\mu_{b_1,b_2,r,1}/\mu_{b_1,b_2,r,0}$ exceeds the global ratio $\overline{\mu}_{r,1}/\overline{\mu}_{r,0}$, then this type is “too rich” in $A=1$ individuals.  A straightforward two-type calculation (keeping all possible $A=0$ matches and enforcing the global ratio) shows that the fair allocation must lose at most
\[
\frac{
\mu_{b_1,b_2,r,1}\overline{\mu}_{r,0}
- \mu_{b_1,b_2,r,0}\overline{\mu}_{r,1}
}{
(\mu_{b_1,b_2,r,1}+\mu_{b_1,b_2,r,0})\,\overline{\mu}_{r,0}
}
\]
of the local matched mass.  When the inequality is reversed, the type is “too rich” in $A=0$, and by symmetry (exchanging the roles of $A=0$ and $A=1$) the maximal local loss is
\[
\frac{
\mu_{b_1,b_2,r,0}\overline{\mu}_{r,1}
- \mu_{b_1,b_2,r,1}\overline{\mu}_{r,0}
}{
(\mu_{b_1,b_2,r,1}+\mu_{b_1,b_2,r,0})\,\overline{\mu}_{r,1}
}.
\]
Taking the larger of these two expressions gives the loss bound for fixed $(b_1,b_2,r)$.

\medskip
\noindent\textbf{Step 2: Global bound across all \((b_1,b_2,r)\).}
At a fixed sensitization level $r$, the total matches in the unconstrained optimum are a sum over $(b_1,b_2)$ of $M^*_{b_1,b_2,r}(N)$, and the number of matches in any fair allocation is the corresponding sum after trimming.  Hence, at level $r$, the (asymptotic) fractional loss is a convex combination of the local losses $\mathrm{Loss}_{b_1,b_2,r}$, and so is bounded by their maximum over $(b_1,b_2)$.

Taking also the maximum over $r$ yields the loss bound in the statement. 

\medskip
\noindent\textbf{Step 3: Implementability in $G_N$ via subgroup prioritization.}
We now argue that the type-level bound $L$ governs the asymptotic price of fairness in the random graphs $G_N$.

For each triple $(b_1,b_2,r)$, the trimming operation in Step 1 can be implemented by prioritizing between the fine types $(b_1,b_2,r,1)$ and $(b_1,b_2,r,0)$, as in Section \ref{subsec::general}: we decide, for each $(b_1,b_2,r)$, whether to prioritize $A=1$ or $A=0$, and collect these choices into a priority set $\mathcal{P}$.

Consider the constrained optimization problem with subgroup prioritization constraints \eqref{eq::theory-priority-obj}–\eqref{eq::theory-priority-con} associated with $\mathcal{P}$.  By Proposition \ref{prop2}, there exists a sequence $\varepsilon_N\downarrow0$ such that we can impose the asymptotic priority requirements for all $(b_1,b_2,i,j)\in\mathcal{P}$ while losing only $o(N)$ matches relative to the unconstrained optimum.

Within this prioritized class, we can then discard matches exactly as in the type-level construction of Step 1 to enforce equality of selection probabilities between $A=0$ and $A=1$ at each $r$.  The additional loss caused by this trimming is, by Step 2, at most a fraction $L$ of the total number of matches, that is, at most $L\,\mathrm{Opt}(G_N)+o(N)$ matches in $G_N$.

Thus, almost surely, there exists an allocation satisfying the stronger condition of equal selection probabilities at each $r$ whose total number of matches is at least $(1-L)\,\mathrm{Opt}(G_N)-o(N)$.  Since the fairness constraint \eqref{eq::fair-0} only requires the difference in selection probabilities at each $r_j$ to be bounded by $l_j$ and not necessarily equal to zero, any optimal solution under \eqref{eq::fair-0} is at least as efficient as this equalized allocation.  Consequently, the asymptotic price of fairness due to \eqref{eq::fair-0} is no greater than $L$.
\end{proof}

\begin{proof}[Proof of Proposition \ref{prop3}]
We work in the random graph model of Section \ref{sec::theory} specialized to the case
where the objective is to maximize the number of transplants (no recourse strategies), and
we adopt the additional structural assumptions stated before Proposition \ref{prop3}.

\medskip
\noindent\textbf{Notational convention.}
In the rest of the paper we index blood types of a pair as
\[
(b_1,b_2) = (\text{donor blood type},\ \text{patient blood type}).
\]
By contrast, Ashlagi and Roth describe a pair of \emph{type} $X$--$Y$ when the \emph{patient}
has blood type $X$ and the \emph{donor} has blood type $Y$.
Thus their type $X$--$Y$ corresponds to our ordered pair
\[
(b_1,b_2) = (Y,X).
\]
Whenever we refer below to $X$--$Y$ \emph{pair types} or to their classification
as self-/over-/under-demanded from \cite{ashlagi2014free}, this is always in the
\emph{patient--donor} sense of Ashlagi--Roth; when we translate those statements
into our formulas in terms of $\mu_{b_1\mid a}$ and $\mu_{b_2\mid a}$, we use the
above correspondence.

\medskip
The proof has four main steps.
\begin{itemize}
\item Step 1 recalls the blood-type structure of an asymptotically optimal allocation
without fairness constraints, following \cite{ashlagi2014free}, and translates it into
our $(b_1,b_2)$ notation.
\item Step 2 quantifies how much matched ``mass'' can be reallocated between $A=0$ and
$A=1$ within partially matched blood-type classes; this yields $R_0$ and $R_1$.
\item Step 3 uses this flexibility and, if necessary, discards some matches to enforce
equality of selection probabilities across $A=0$ and $A=1$ at each sensitization
level $r$, and computes the resulting loss.
\item Step 4 uses Proposition \ref{prop2} to argue that this type-level construction can
be implemented in the random graphs $G_N$ up to $o(N)$ loss and that the weaker
fairness constraint \eqref{eq::fair-0} has no higher asymptotic price.
\end{itemize}

\medskip
\noindent\textbf{Step 1: Structure of the unconstrained optimal allocation.}
Fix $N$ and let $G_N$ be the random compatibility graph on $N$ incompatible pairs.
For each $a\in\{0,1\}$ and $r\in\{r_1,\dots,r_M\}$, write
\[
V_{a,r} := \{v\in V_N : A(v)=a,\ R(v)=r\}.
\]
By the strong law and the multiplicative structure,
\[
\frac{|V_{a,r}|}{N}
\;\longrightarrow\;
\sum_{b_1,b_2} \mu_{b_1,b_2,r,a}
= c\,r\,\mu_a
\quad\text{a.s.\ as }N\to\infty.
\]

In \cite{ashlagi2014free}, pair types are indexed as $X$--$Y$ with $X$ the \emph{patient}
blood type and $Y$ the \emph{donor} blood type. Translating into our $(b_1,b_2)$ notation
via $(b_1,b_2)=(Y,X)$, a type $X$--$Y$ in \cite{ashlagi2014free} corresponds to
\emph{patient type $X$ and donor type $Y$}, i.e.\ our $(b_1,b_2)=(Y,X)$.

Using that convention, \cite{ashlagi2014free} partition the $4\times4$ blood-type
pairs $X$--$Y$ as follows:
\begin{itemize}
  \item self-demanded: $X$--$X$ (patient and donor have the same blood type);
  \item over-demanded: $X$--$Y$ with $Y\to X$ and $X\neq Y$ (donor offers a more
        flexible blood type than the patient needs);
  \item under-demanded: $X$--$Y$ with $X\to Y$ and $X\neq Y$ (patient seeks a more
        flexible type than the donor offers);
  \item reciprocally demanded: $A$--$B$ and $B$--$A$.
\end{itemize}

Under the numerical assumptions on blood-type frequencies and sensitization, 
they show that in almost every large pool there exists an optimal allocation
(with exchanges of length at most $3$) whose limiting structure is deterministic
in terms of these blood-type classes. After translating their $X$--$Y$ types into
our $(b_1,b_2)$ notation, the key consequences \emph{within each group}
$(A=a,R=r)$ are:
\begin{itemize}
  \item all self-demanded types are asymptotically fully matched;
  \item all $A$--$B$ and $B$--$A$ \emph{patient--donor} types are asymptotically fully
        matched via 2-cycles and 3-cycles;
  \item all AB--O \emph{patient--donor} types are matched
        using 3-cycles involving overdemanded types;
  \item all remaining overdemanded types are matched to their corresponding
        underdemanded types in 2-cycles;
  \item only the genuinely underdemanded types beyond available overdemanded partners
        remain unmatched.
\end{itemize}

Counting these matched classes in our $(b_1,b_2)$ notation, and using the product form
\[
\mu_{b_1,b_2,r,a}
= c\,r\,\mu_a\,\mu_{b_1\mid a}\mu_{b_2\mid a},
\]
gives exactly the expressions $T_a(r)$ and $S_a(r)$ stated before the proposition.
In particular, for each $a$ and $r$ the unconstrained selection probability in
group $(A=a,R=r)$ converges to
\[
p_a(r) := \frac{T_a(r)}{S_a(r)}.
\]
Aggregating across $a$ gives the total matchable mass at level $r$, proportional
to $Q(r)$ as defined in the statement.

\medskip
\noindent\textbf{Step 2: Reallocatable mass between groups and $R_a$.}
The Ashlagi--Roth construction leaves some blood-type classes
(particularly those involving underdemanded types) only partially utilized. Within
such classes we can reassign which fraction of matched pairs comes from $A=0$ versus
$A=1$, provided we respect:
\begin{itemize}
  \item for each blood-type class, the total matched mass in that class; and
  \item for each group $A=a$, the total number of available pairs of that class.
\end{itemize}

For each ``bottleneck'' class (e.g.\ patient--donor $B$--AB, $O$--AB, $A$--AB, $A$--O,
AB--O, $O$--B in the Ashlagi--Roth notation) this gives a small linear program in two
variables (mass from $A=0$ and from $A=1$). Solving these programs and translating
back to our $(b_1,b_2)$ representation gives the quantities $R_a$ as in the proposition:
\[
\begin{aligned}
R_a=&\min\left\{\phi_{B,AB},2\mu_a\mu_{B\mid a}\mu_{AB\mid a}-\phi_{B,AB}\right\}
+\min\left\{\phi_{O,AB}+\phi_{A,AB},2\mu_a\mu_{A\mid a}\mu_{AB\mid a}-\phi_{O,AB}-\phi_{A,AB}\right\}\\
&+\min\left\{\phi_{A,O}+\phi_{AB,O},2\mu_a\mu_{O\mid a}\mu_{A\mid a}-\phi_{A,O}-\phi_{AB,O}\right\}
+\min\left\{\phi_{O,B},2\mu_a\mu_{O\mid a}\mu_{B\mid a}-\phi_{O,B}\right\}.
\end{aligned}
\]
Informally, $R_a$ is the maximal amount of matched mass that can be shifted in
favour of group $A=a$ by reassigning responsibility for matches within these partially
utilized blood-type classes.

\medskip
\noindent\textbf{Step 3: Enforcing equality of selection probabilities at each $r$.}
Fix a sensitization level $r$. We aim for the stronger condition
\[
\frac{\text{\# matches in }(A=1,R=r)}{|V_{1,r}|}
=
\frac{\text{\# matches in }(A=0,R=r)}{|V_{0,r}|},
\]
i.e.\ $p_1(r)$ and $p_0(r)$ are made equal. This is stricter than \eqref{eq::fair-0},
so any loss bound obtained here is valid for our weaker fairness constraint.

Suppose first that $p_1(r) < p_0(r)$, so group $A=1$ is disadvantaged.

\smallskip
\emph{(a) Using $R_1$ to help $A=1$.}
By reallocating within partially matched classes we can increase the matched mass
of $(A=1,R=r)$ by at most $R_1$, without changing the total mass at level $r$.
This yields a new numerator $T_1(r)+R_1$ for $A=1$ and reduces the numerator
for $A=0$ by the same $R_1$.

If this already suffices for equality,
\[
\frac{T_1(r)+R_1}{S_1(r)} \;\ge\; \frac{T_0(r)}{S_0(r)},
\]
then equality can be achieved with no loss of total matches at level $r$.

\smallskip
\emph{(b) If reallocation is insufficient, discarding from $A=0$.}
Otherwise we must additionally discard matches from $(A=0,R=r)$.
Let $x(r)\ge0$ be the extra amount of matched mass we remove from group $A=0$
at level $r$, beyond the $R_1$ already reallocated. Equality of selection probabilities
requires
\[
\frac{T_1(r)+R_1}{S_1(r)} \;=\; \frac{T_0(r)-x(r)}{S_0(r)},
\]
hence
\[
x(r)
= \frac{S_0(r)T_1(r)-S_1(r)T_0(r)-S_0(r)R_1}{S_1(r)}.
\]

\smallskip
\emph{(c) Symmetric case $p_1(r)\ge p_0(r)$.}
If $p_1(r)\ge p_0(r)$, the roles of the two groups are reversed: we can first
reallocate up to $R_0$ in favour of $A=0$ and, if needed, discard an additional
amount $x'(r)\ge0$ from group $A=1$, where
\[
x'(r)
= \frac{S_1(r)T_0(r)-S_0(r)T_1(r)-S_1(r)R_0}{S_0(r)}.
\]

\smallskip
\emph{(d) Relative loss at level $r$.}
At level $r$, the total matchable mass in the unconstrained optimum is
proportional to $Q(r)$; hence the relative loss (price of fairness) at level
$r$ is at most
\[
\max\left\{
\frac{S_1(r)T_0(r)-S_0(r)T_1(r)-S_0(r)R_1}{S_1(r)Q(r)},\;
\frac{S_0(r)T_1(r)-S_1(r)T_0(r)-S_1(r)R_0}{S_0(r)Q(r)},\;
0
\right\},
\]
which is exactly the quantity appearing in the statement.

\medskip
\noindent\textbf{Step 4: Implementability via subgroup prioritization and Proposition \ref{prop2}.}
Up to now the construction has been at the level of limiting type proportions. We now
argue that it can be implemented in the random graphs $G_N$ (with exchanges of bounded
length) up to $o(N)$ loss.

For each triple $(b_1,b_2,r)$ and $a\in\{0,1\}$, the reallocation and discarding
operations in Step 3 determine whether, at that type, the fair allocation effectively
\emph{prioritizes} group $A=1$ or $A=0$. Collect these choices in a priority set
$\mathcal{P}$ as in Section \ref{subsec::general}, and consider the optimization
problem with subgroup-prioritization constraints
\eqref{eq::theory-priority-obj}–\eqref{eq::theory-priority-con} associated with
$\mathcal{P}$.

By Proposition \ref{prop2}, there exists a sequence $\varepsilon_N\downarrow0$ such that,
almost surely, we can impose the asymptotic priority requirements
\eqref{eq::theory-priority-con} for all $(b_1,b_2,i,j)\in\mathcal{P}$ while losing only
$o(N)$ matches relative to the unconstrained optimum. In other words, the prioritized
solutions implement, up to $o(N)$, the type-level pattern constructed in Step 3.

Within this prioritized class, we can then carry out the same trimming as in Step 3
to enforce equality of selection probabilities at each $r$. The additional loss in the
number of matches is, at level $r$, bounded by the relative loss above, and hence the
overall asymptotic price of fairness is bounded by
\[
\max_{r\in\{r_1,\dots,r_M\}}
\max\left\{
\frac{S_1(r)T_0(r)-S_0(r)T_1(r)-S_0(r)R_1}{S_1(r)Q(r)},\;
\frac{S_0(r)T_1(r)-S_1(r)T_0(r)-S_1(r)R_0}{S_0(r)Q(r)},\;
0
\right\}.
\]

Finally, the original fairness constraint \eqref{eq::fair-0} requires only that the
difference in selection probabilities at each $r_j$ be bounded by $l_j$, not that
they be exactly equal. Thus an optimal solution under \eqref{eq::fair-0} is at least
as efficient as the equalized allocations we have constructed, so its asymptotic
price of fairness is no larger than the bound stated in Proposition \ref{prop3}.

This completes the proof.
\end{proof}

\section{More simulations with failure-aware strategies}\label{sec::appendix-simulation}

We consider a similar data-generating process in Section \ref{subsec::random_graph} but with $N=50$, and we repeat exactly the same data-generating process in Section \ref{subsec::unos}. With additional vertex and edge uncertainties, let the failure probability $p_v$ and $p_{v_i,v_j}$ be independently sampled from a uniform distribution $U(0, 0.3)$. We consider the subset-recourse strategy and choose $\mathcal{S}$ to be the set of relevant subsets of size $(3,1)$.

Figures \ref{fig-s1} and \ref{fig-s2} present the average selection rates of different fairness criteria within each subgroup. The numerical performance with failure-aware strategies is similar to that without failure-aware strategies.

\begin{figure}[h]
\caption{Simulation results with subset-recourse strategy under random graph models. The average selection rates within each subgroup are calculated over 100 data replications. The error bars represent the mean $\pm$ 1 standard deviation of the absolute differences in selection rates. }\label{fig-s1}
\centering
\includegraphics[width=0.8\textwidth]{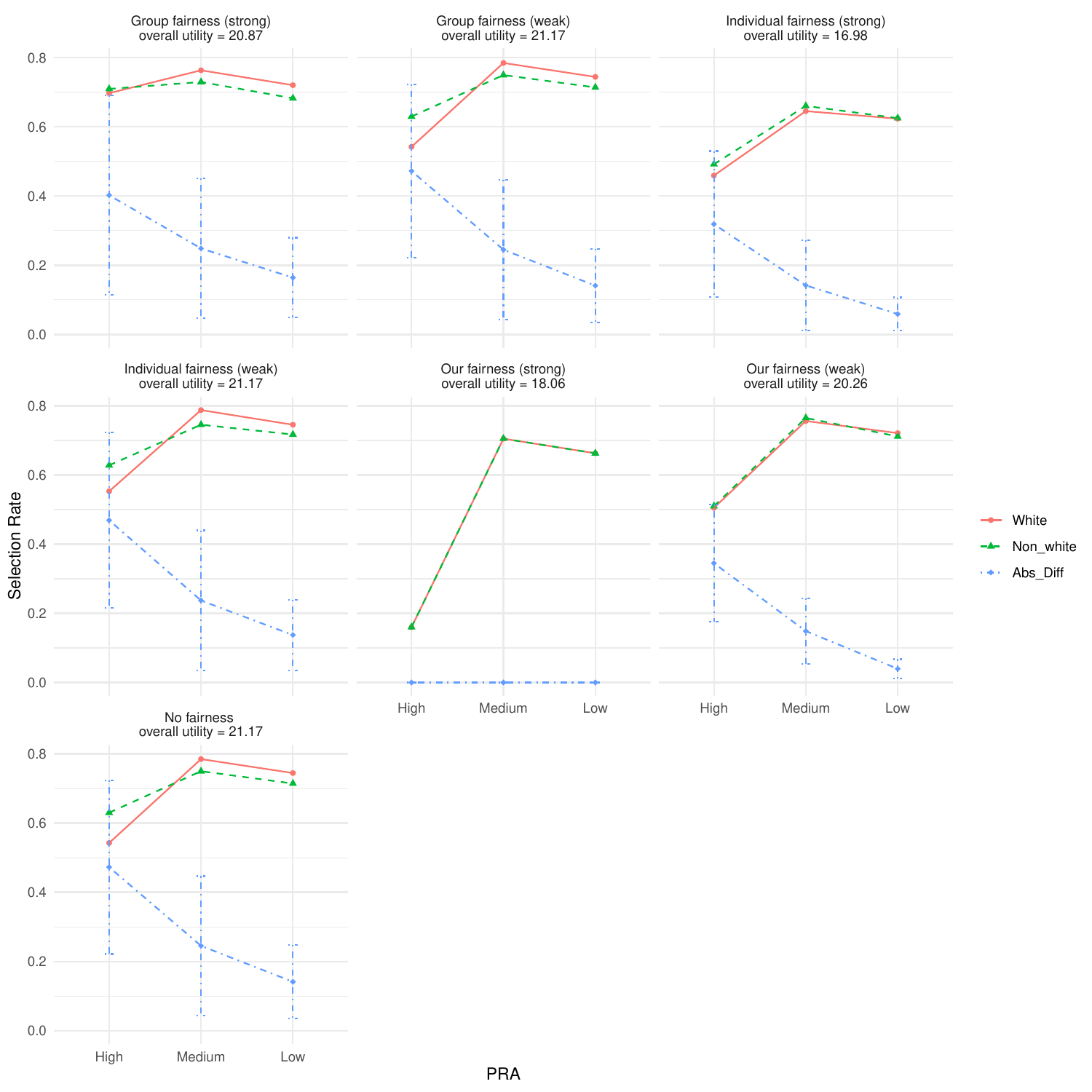}
\end{figure}

\begin{figure}[h]
\caption{Simulation results with subset-recourse strategy based on UNOS data. The average selection rates within each subgroup are calculated over 100 data replications. The error bars represent the mean $\pm$ 1 standard deviation of the absolute differences in selection rates. }\label{fig-s2}
\centering
\includegraphics[width=0.8\textwidth]{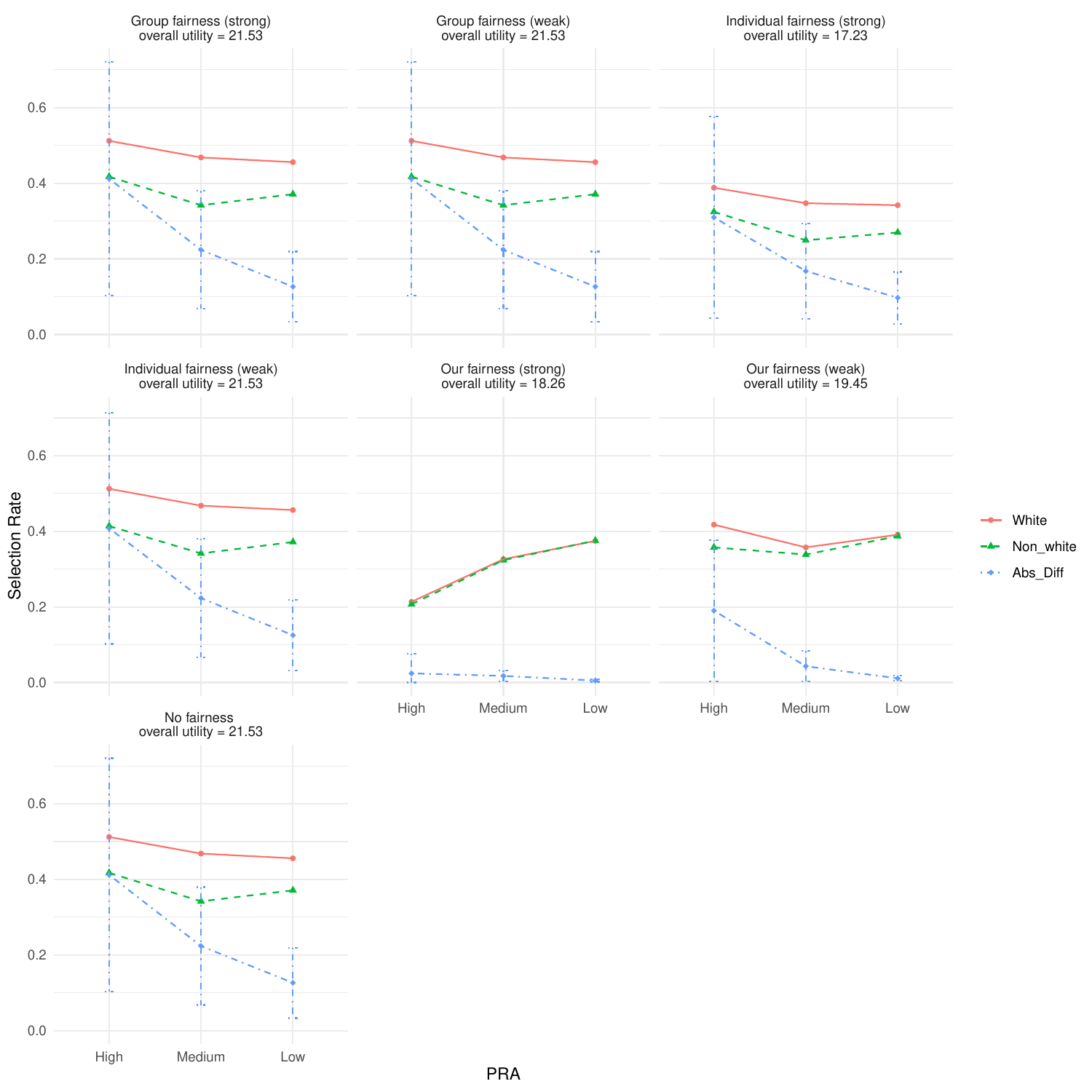}
\end{figure}

\end{document}